\documentclass[letterpaper,journal,romanappendices]{IEEEtran}
\usepackage{amsmath,amssymb,amsfonts,amsthm}
\usepackage{algorithmic}
\usepackage{algorithm}
\usepackage{array}
\usepackage{textcomp}
\usepackage{stfloats}
\usepackage{url}
\usepackage{verbatim}
\usepackage{graphicx}
\usepackage[caption=false,font=footnotesize,labelfont=sf,textfont=sf]{subfig}
\usepackage{cite}
\usepackage{xcolor}
\usepackage{braket}
\usepackage[hidelinks]{hyperref}
\usepackage[T1]{fontenc}
\usepackage{cleveref}
\usepackage{gensymb}
\usepackage{orcidlink}
\usepackage{braket}
\usepackage{cleveref}
\usepackage{tabularx} 
\def\BibTeX{{\rm B\kern-.05em{\sc i\kern-.025em b}\kern-.08em
    T\kern-.1667em\lower.7ex\hbox{E}\kern-.125emX}}

\newcommand{\nb}{\bar{n}_\mathrm{B}}
\newcommand{\ns}{\bar{n}_\mathrm{s}}
\newcommand{\im}{j}

\DeclareMathOperator{\tr}{tr}
\newcommand{\ketbra}[2]{\ket{#1}\bra{#2}}

\newcommand{\pfail}{p_\text{f}}
\newcommand{\bobiqubit}{\hat{\tau}_{R\check{B}_i}}

\theoremstyle{definition}
\newtheorem{definition}{Definition}[section]

\newtheorem{theorem}{Theorem}
\newtheorem{lemma}{Lemma}

\theoremstyle{remark} \newtheorem{remark}{Remark}

\allowdisplaybreaks

\renewcommand{\thesubsubsection}{\arabic{subsubsection}}

\begin{document}

\title{Covert Entanglement Generation over Bosonic Channels
\thanks{This paper was presented in part at the IEEE International Conference on Quantum Computing and Engineering (QCE) in Montr\'{e}al, QC, Canada, September 2024 \cite{anderson2024covert-qce} and at the IEEE International Symposium on Information Theory (ISIT) in Ann Arbor, MI, USA, June 2025 \cite{anderson2025covert-isit}.}
\thanks{This work was supported by the National Science Foundation under Grants No. CCF-2006679 and EEC-1941583,  
 the Israel Science Foundation under Grants No. 939/23 and 2691/23, German-Israeli Project Cooperation (DIP) within the Deutsche Forschungsgemeinschaft (DFG) under Grant No.
2032991, Ollendorff Minerva Center (OMC) of the Technion No. 86160946, and by the Junior Faculty Program
for Quantum Science and Technology of Israel Planning and Budgeting Committee of the Council for Higher Education (VATAT) under Grant No. 86636903. 

}}

\author{Evan J.~D.~Anderson$^*$\orcidlink{0000-0002-0894-1695}, Michael S.~Bullock$^*$\orcidlink{0000-0002-3528-7473}, Ohad Kimelfeld\orcidlink{0009-0007-8398-0719}, Christopher K. Eyre\orcidlink{0009-0006-8755-8868}, Filip Rozp\k{e}dek\orcidlink{0000-0002-2755-4623},\\ Uzi Pereg\orcidlink{0000-0002-3259-6094}, and Boulat A.~Bash\orcidlink{0000-0002-1205-3906}
\thanks{*These authors contributed equally to this work.}
\thanks{Evan Anderson is with the Wyant College of Optical Sciences, University of Arizona
Tucson, AZ, USA
(email: ejdanderson@arizona.edu)}
\thanks{Michael Bullock is with Dept.~of Electrical and Computer Engineering, University of Arizona, Tucson, AZ, USA (email: bullockm@arizona.edu) }
\thanks{Ohad Kimelfeld is with the Physics Department and Helen Diller Quantum Center, Technion -- Israel Institute of Technology, Haifa, Israel (email: ohad.kim@campus.technion.ac.il) }
\thanks{Christopher Eyre is with the Dept.~of Mathematics, Brigham Young University
Provo, UT, USA }
\thanks{Filip Rozp\k{e}dek is with the College of Information \& Computer Sciences, University of Massachusetts, Amherst, MA, USA (email: frozpedek@cics.umass.edu)}
\thanks{Uzi Pereg is with the Electrical and Computer Engineering Department and Helen Diller Quantum Center, Technion -- Israel Institute of Technology, Haifa, Israel (email: uzipereg@technion.ac.il)}
\thanks{Boulat Bash is with Dept.~of Electrical and Computer Engineering, and the Wyant College of Optical Sciences, University of Arizona, Tucson, AZ, USA (email: boulat@arizona.edu)}}

\maketitle

\begin{abstract}
We explore covert entanglement generation over the lossy thermal-noise bosonic channel, which models quantum-mechanically  many practical channels, including optical, microwave, and radio-frequency (RF). Covert communication ensures that an adversary cannot detect the presence of transmissions, which are concealed in channel noise. 
We show a \emph{square root law} (SRL) for covert entanglement generation similar to that for classical communication: $L_{\rm EG}\sqrt{n}$ entangled bits (ebits) can be generated covertly and reliably over $n$ uses of a bosonic channel.  We report a single-letter expression for optimal $L_{\rm EG}$ as well as an achievable method. We additionally analyze the performance of covert entanglement generation using single- and dual-rail photonic qubits, which may be more practical for physical implementation.
\end{abstract}

\section{Introduction}

Standard communication security protects the transmission content from unauthorized access using cryptography \cite{menezes96HAC} or information-theoretic secrecy \cite{bloch11pls}.
On the other hand, covert, or low-probability-of-detection/intercept (LPD/LPI), signaling prevents detection of the transmission in the first place. Over the last decade, the fundamental limits of covert communication were explored for classical  \cite{bash12sqrtlawisit, bash13squarerootjsac,bloch15covert, wang15covert} and classical-quantum channels \cite{bash15covertbosoniccomm, bullock20discretemod, gagatsos20codingcovcomm,  azadeh16quantumcovert-isitarxiv, bullock2025fundamentallimitscovertcommunication, bullockCovertCommunicationClassicalQuantum2023, zlotnick25eacqcovert}. 
Covert communication over these channels is governed by the \textit{square root law} (SRL): $L_{\rm c}\sqrt{n}$ covert bits are reliably transmissible over $n$ channel uses for a channel-dependent constant $L_{\rm c}>0$ called covert classical channel capacity, and has units $\text{bit}/\sqrt{\text{channel use}}$. In the typical case when $n=TW$, the time-bandwidth product of a transmission, $L_{\rm c}$ is measured in $\text{b}/\sqrt{\text{s}\times\text{Hz}}$.
A tutorial \cite{bash15covertcommmag} and a detailed survey \cite{chen23covcommssurvey} provide an overview of these results and their developments.

To date, most covert communication efforts have focused on classical data, i.e., transmission of bits. 
Motivated by the importance of quantum communication discussed next, in this paper, we explore its covert variant. Specifically, we explore covert entanglement generation, the process of establishing entanglement between Alice and Bob while remaining undetected by an adversarial quantum-powerful warden named Willie.
Here, we investigate covert entanglement generation over lossy thermal-noise bosonic channels, which we call ``bosonic channels'' for brevity and formally define in Section \ref{sec:systemmodel}. 
They provide a quantum-mechanical description of optical-fiber, free-space-optical (FSO), microwave, and radio-frequency (RF) communication channels. 
Notably, unlike standard qubits, which are finite-dimensional, the quantum states of light on which these channels act are infinite-dimensional. 

\subsection{Motivation}
\label{sec:motivation}
Quantum information processing \cite{nielsen00quantum, wilde16quantumit2ed} offers a significant advantage over classical methods by harnessing the quantum-mechanical principles of superposition and entanglement.
For example, in sensing, these allow substantial enhancement of measurement precision 
in interferometry, magnetometry, and timing applications 
\cite{degen17quantumsensing, pirandola18quantumsensing}. In computing, these can achieve exponential speedup in quantum algorithms such as Shor's algorithm \cite[Ch. 5]{nielsen00quantum}. In communication, entanglement enables quantum key distribution (QKD) \cite{E91}, as well as increases classical capacity \cite[Ch. 21]{wilde16quantumit2ed} \cite{giovannetti03broadband,bennettEntanglementAssistedClassicalCapacity1999}.
In fact, entanglement assistance improves on the SRL in covert communication, allowing reliable transmission of $\propto\sqrt{n}\log{n}$ rather than $\propto \sqrt{n}$ covert bits \cite{gagatsos20codingcovcomm, zlotnick25eacqcovert}. 

Quantum-enabled devices derive these advantages by acting on quantum states directly.
Quantum communication \cite{nielsen00quantum, wilde16quantumit2ed, wehner18quantuminternet} enables distributing quantum operations spatially, further increasing their capability.
Quantum information theory \cite{wilde16quantumit2ed} generalizes classical \cite{cover02IT} by including an underlying quantum-mechanical description of a physical communication channel.
We illustrate the power of quantum communication by considering the basic channel that transmits qubits, or two-dimensional quantum systems, from Alice to Bob as in \cite[Ch.~24]{wilde16quantumit2ed}.
By Holevo's theorem \cite{holevo73bound}, \cite[Ch.~11] {wilde16quantumit2ed}, at most $\propto n$ classical bits can be transmitted reliably via $n$ uses of this channel by Alice encoding them in $n$ qubits and Bob measuring his output. 
Now, suppose Alice must transmit an arbitrary $n$-qubit quantum state.
First, consider transmitting its classical description.
Alice needs to perform quantum state tomography by measuring $\propto 4^n$ copies of $n$ qubits, record $\propto 4^n$ classical bits that characterize their quantum state \cite{jamesMeasurementQubits2001}, and transmit them so that Bob can reconstruct it.
This is complicated because these $n$ qubits can be in an arbitrary superposition and entangled.
Furthermore, the no-cloning theorem of quantum mechanics \cite[Sec.~3.5.4]{wilde16quantumit2ed} precludes replicating unknown quantum states, rendering tomography ineffective when only a few copies are available, such as in many quantum sensing scenarios.
Direct transmission of quantum states, thus, is not only a convenience but also a necessity.
Quantum teleportation enables this through quantum entanglement and classical communication between Alice and Bob: to teleport the aforementioned $n$-qubit quantum state, Alice and Bob need to generate $n$ maximally-entangled qubit pairs (known as ``entangled bits'' or ebits) and Alice needs to transmit $2n$ classical bits \cite[Sec.~1.3.7]{nielsen00quantum} \cite[Sec.~6.5.3]{wilde16quantumit2ed}.

Thus, entanglement is a cornerstone of quantum information processing. 
Enabling covert quantum computing and sensing and improving covert classical communication therefore motivate our investigation of covert entanglement distribution.

\subsection{Prior Work, Initial Conference Results, and Research Gap} \label{sec:prior-work}
Given the progress in covert classical communication \cite{bash12sqrtlawisit, bash13squarerootjsac,bloch15covert, wang15covert, bash15covertbosoniccomm, bullock20discretemod, gagatsos20codingcovcomm, azadeh16quantumcovert-isitarxiv, bullock2025fundamentallimitscovertcommunication, bullockCovertCommunicationClassicalQuantum2023, zlotnick25eacqcovert, bash15covertcommmag, chen23covcommssurvey},  the extension to covert quantum communication is natural.
The initial study \cite{arrazolaCovertQuantumCommunication2016}, as well as the follow-on work \cite{tahmasbi19covertqkd,tahmasbi20bosoniccovertqkd-jsait, tahmasbi20covertqkd}, focus on covert QKD. That is, \cite{arrazolaCovertQuantumCommunication2016, tahmasbi19covertqkd, tahmasbi20bosoniccovertqkd-jsait, tahmasbi20covertqkd} seek to generate a classical secret between Alice and Bob covertly using classical and quantum channel resources, analogous to a well-known non-covert approach \cite{E91}.

Although the SRL for direct covert transmission of quantum states is discussed in \cite{arrazolaCovertQuantumCommunication2016}, covert quantum communication remained underexplored until our conference papers \cite{anderson2024covert-qce, anderson2025covert-isit}.
There, we investigate quantum communication between parties that already possess covert classical resources, such as a pre-shared classical secret and a covert classical channel. 
In \cite{anderson2024covert-qce}, we confirm the SRL and derive the lower (achievable) and upper (converse) bounds on the covert quantum capacity of bosonic channels.
Although we improve the lower bound in \cite{anderson2025covert-isit}, it does not match the converse in \cite{anderson2024covert-qce}, which we posit is too loose.

While QKD protocols \cite{E91} enable generating classical secrecy from shared entanglement, conversely, shared classical secret can be used for generating entanglement \cite{devetakPrivateClassicalCapacity2005, hayden2008decoupling}.
Quantum states held by Alice and Bob that are maximally entangled are uncorrelated with any other quantum states, including those in the environment (in fact, entanglement decoherence refers to correlation ``leakage'' to the environment). Alice and Bob's classical shared secret is similarly uncorrelated with any state, classical or quantum, held by the adversary. Treating the environment as an adversary, Alice and Bob can use a classical secret to construct a quantum code that allows them to employ a quantum channel to generate entanglement \cite{devetakPrivateClassicalCapacity2005, hayden2008decoupling}. 
This, combined with prior work on covert classical communication over finite-dimensional classical \cite{bloch15covert} and quantum channels \cite{azadeh16quantumcovert-isitarxiv, bullock2025fundamentallimitscovertcommunication, anderson2025covert-isit} inspires a companion paper \cite{kimelfeld25ceg}  to this article  that takes a different approach from \cite{anderson2024covert-qce, anderson2025covert-isit} and derives the fundamental limit of covert entanglement generation over finite-dimensional quantum channels.
We note that the achievability and converse proofs in \cite{kimelfeld25ceg} yield identical upper and lower bounds, thus completely characterizing the covert entanglement generation capacity of these channels.
Such covertly generated entanglement can be used with a covert classical channel to teleport quantum states between Alice and Bob, as described in Section \ref{sec:motivation}.

While finite-dimensional qubit channels have many applications (most notably in quantum computing), bosonic channels that model many physical channels are infinite-dimensional.
The corresponding dichotomy between discrete and  continuous variables often leads to dramatically different results.
For example, in classical deterministic identification, the scaling is $\propto 2^{nR}$ for discrete variables, and $\propto 2^{R n\log n }=n^{nR}$ for continuous variables \cite{salariseddigh22ID}.
Another example is the dramatic difference between quantum error correction techniques for qubit \cite[Ch.~10]{nielsen00quantum} \cite{gottesmanIntroductionQuantumError2009} and bosonic channels  \cite{gottesmanEncodingQubitOscillator2001}.
Thus, the results in \cite{kimelfeld25ceg} cannot simply be applied to our scenario and must be re-derived for the bosonic channels.
Furthermore, current and near-term quantum channel codes and modulation schemes must be analyzed for supporting covertness.

\subsection{Our Contributions}\label{sec:intro-contributions}

We report an SRL similar to classical communication for covert entanglement generation over the bosonic channels: $L_{\rm EG}\sqrt{n}$ covert entangled qubit pairs (ebits) can be generated over $n$ channel uses for a channel-dependent constant $L_{\rm EG}>0$ called covert entanglement-generation capacity. 
We characterize the optimal value by providing 1) an upper bound on $L_{\rm EG}$, and 2) a method to generate $L_{\rm EG}\sqrt{n}$ covert ebits using a classical secret that is pre-shared between communicating parties. 
Our approach adapts the covert entanglement generation method for finite-dimensional channels from \cite{kimelfeld25ceg} to infinite-dimensional bosonic channels.
This non-trivial derivation builds on the results in \cite{bullock20discretemod, gagatsos20codingcovcomm, wang22isitcoverttd, wang22isittowardtdjournal}.
In fact, the covert entanglement generation capacity $L_{\rm EG}=L_{\text{no-EA}}$, the unassisted covert classical capacity derived in \cite{gagatsos20codingcovcomm}, which agrees with the finite-dimensional result in \cite{kimelfeld25ceg}.
We note that, although $L_{\rm EG}$ and $L_{\text{no-EA}}$ are numerically the same, the units for the former are $\text{ebit}/\sqrt{\text{channel use}}$, while, for the latter, they are $\text{bit}/\sqrt{\text{channel use}}$. This reflects the operational differences between classical and quantum transmission described in Section \ref{sec:motivation}. 
Concurrent with covert entanglement generation, we obtain an optimal-rate covert and information-theoretically secure classical communication scheme.
This prevents leakage of information contained in the transmission in the rare event that it is detected.

Remarkably, our formula for $L_{\rm EG}$ has a single-letter form, while only the bounds for the non-covert quantum capacity of a bosonic channel are currently known \cite{pirandolaFundamentalLimitsRepeaterless2017, noh2018capacitybounds, noh20energyconstrainedquantumcomm, fanizza21quantumcap} 
(The importance of single-letterization in quantum information theory is highlighted in \cite[Section VI.B]{Pereg2023pCommunication}).
The non-covert entanglement generation capacity upper bounds the quantum capacity because entanglement and classical communication enable teleportation of quantum states.
However, in covert quantum communication, the classical channel must also be covert, necessitating careful analysis, which we defer to future work.

Furthermore, the complexity of our optimal entanglement-generation scheme makes its realization using known components extremely challenging.
Thus, we explore covert entanglement generation using sub-optimal but more practical single- and dual-rail photonic qubits.
These qubit modulation schemes are used in many quantum information processing tasks, including cluster-state generation \cite{thomasEfficientGenerationEntangled2022}, entanglement distribution in quantum networks \cite{takedaDeterministicQuantumTeleportation2013, guhaRateLoss2015, hensenLoopholefreeBellInequality2015, krutyanskiyEntanglementTrappedIonQubits2023,   dhara2023entangling, azuma2023repeatersurvey}, and quantum key distribution (QKD) \cite{Honjo08qkd, scarani09rmpQKD}.
We report a significant gap relative to our optimal scheme, motivating further investigation.

The rest of this paper is organized as follows: next, we state the mathematical preliminaries as well as the system and channel models. In Section~\ref{sec:primary-results}, we provide the ultimate limits of covert entanglement generation. In Section \ref{sec:practical}, we investigate covert entanglement generation using single- and dual-rail photonic qubit modulation. In Section~\ref{sec:performance}, we compare the performance of our methods. We conclude in Section \ref{sec:conclusion} by discussing the implications of our results and areas of future research.  Appendices contain our mathematical proofs.

\section{Preliminaries}
\label{sec:preliminaries}

\subsection{Notation} \label{sec:notation}
\subsubsection{Linear algebra and quantum mechanics}
We employ standard notation for quantum information processing, as found in \cite{wilde16quantumit2ed}, \cite[Ch.~2.2.1]{tomamichel15finiteresourcesQIP}, and \cite{hayashi2006quantum}. We use hats for operators, e.g., $\hat{X}$ is the Pauli-X operator, and $\hat{\rho}_{A^n}$ is the $n$-mode normalized density operator (matrix) that belongs to system $A$. Similarly, $\ket{\cdot}$ is a column vector representing a quantum state, while $\bra{\cdot}$ is its conjugate transpose. These vectors live in a Hilbert space of finite or infinite dimension, and the distinction is made explicit or by context.

\subsubsection{Random variables}
We denote scalar random variables and their outcomes in upper- and lower-case, respectively (e.g., $X$ and $x$). We denote random vectors and their realizations in bold upper- and lower-case letters (e.g., $\mathbf{X}$ and $\mathbf{x}$).

\subsubsection{Asymptotics}
We use the standard asymptotic notation \cite[Ch. 3.1]{clrs2e}, where 
\begin{align}
\mathcal{O}(g(n)) &\triangleq \{f(n) : \exists m,n_0 >0 \notag \\ &\phantom{\triangleq } \text{ s.t. } 0\leq f(n) \leq m g(n) \text{ } \forall n \geq n_0\}\\
&=\left\{f(n) : \limsup_{n\to\infty}\left|\frac{f(n)}{g(n)}\right| <\infty\right\}\\
o(g(n))&\triangleq \{f(n) : \forall m>0 \text{,	 } \exists n_0 > 0 \notag \\ &\phantom{\triangleq }\text{ s.t. } 0\leq f(n) < m  g(n)\text{ } \forall n \geq n_0 \}\\
&=\left\{f(n) : \lim_{n\to\infty}\frac{f(n)}{g(n)} =0\right\}\\
\omega(g(n))&\triangleq \{f(n) : \forall m>0, \exists n_0 >0 \notag \\ &\phantom{\triangleq } \text{ s.t. }0\leq m g(n) < f(n)\text{ }\forall	n \geq n_0\}\\
&=\left\{f(n) : \lim_{n\to\infty}\left|\frac{f(n)}{g(n)}\right| = \infty\right\}.
\end{align}
Thus, $\mathcal{O}(g(n))$ denotes an asymptotically tight upper bound on $g(n)$, while $o(g(n))$ and $\omega(g(n))$ are non-asymptotically tight upper and lower bounds on $g(n)$, respectively.

\subsection{Channel and System Model}
\label{sec:systemmodel} 

\begin{figure}[htb]
\vspace{-0.15in}
\centering
\subfloat[]{%
    \includegraphics[width=0.42\columnwidth]{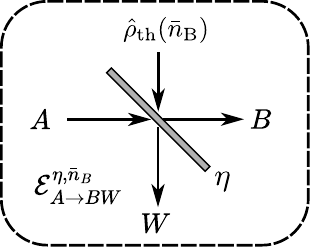}
    \label{fig:thermal-channel}
}
\hfill
\subfloat[]{%
    \includegraphics[width=0.38\columnwidth]{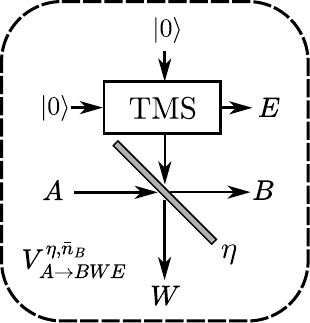}
    \label{fig:iso-channel}
}
\caption{Bosonic channel model. a) is the lossy thermal-noise bosonic channel $\mathcal{E}_{A\to BW}^{\eta,\nb}$, parameterized by $\eta\in[0,1]$ representing loss, and thermal state $\hat{\rho}_{\rm th}(\nb)$ with mean thermal photon number $\nb$ representing the channel noise. The channel has input subsystem at Alice $A$ and output subsystems $B$ and $W$ at Bob and Willie, respectively.  b) is the isometry $V_{A \to BWE}^{\eta,\nb}$ in the Stinespring dilation of the bosonic channel $\mathcal{E}_{A\to BW}^{\eta,\nb}$. The TMS block is a two-mode squeezer with gain $G=1+\bar{n}_B$ and $\ket{0}$ is a vacuum state. }
\label{fig:therm-iso-channels}
\end{figure}

Consider the \emph{lossy thermal-noise bosonic channel}, $\mathcal{E}_{A\to BW}^{\eta,\bar{n}_\textrm{B}}$ in  Fig.~\ref{fig:therm-iso-channels}\protect\subref{fig:thermal-channel}, which we call the bosonic channel for brevity.
It is described by a beamsplitter with transmittance $\eta \in [0,1]$, two input modes (transmitter Alice and thermal environment), and two output modes (legitimate receiver Bob and adversary warden Willie), where a mode is a quantum system corresponding to a single channel use. The parameter $\eta$ models channel photon loss.
Thermal-state input $\hat{\rho}_{\textrm{th}}({\nb})$ at the environment represents additive thermal noise. It has a photon-number (Fock) basis representation  $\hat{\rho}_{\textrm{th}}({\nb}) \equiv \sum_{k=0}^{\infty} \frac{\bar{n}_{\textrm{B}}^k}{(1+\bar{n}_{\textrm{B}})^{k+1}} |k\rangle\langle k|$, where $\nb$ is the mean thermal-environment photon number.
When Alice inputs state $\hat{\rho}_{A^n}$, 
Bob and Willie receive $\hat{\rho}_{B^n}\equiv\tr_{W
^n}\left( \mathcal{E}^{(\eta,\nb)\otimes n}_{A\to BW}(\hat{\rho}_{A^n})\right)$ and $\hat{\rho}_{W^n}\equiv\tr_{B^n}\left( \mathcal{E}^{(\eta,\nb)\otimes n}_{A\to BW}(\hat{\rho}_{A^n})\right)$, respectively. 
Isometric extension or Stinespring dilation represents a quantum channel as an isometry into a larger system containing the ancillary subsystem, followed by discarding this extra subsystem \cite[Ch.~5.2]{wilde16quantumit2ed}. 
Stinespring dilation of the bosonic channel $\mathcal{E}^{(\eta,\nb)\otimes n}_{A\to BW}$ in Fig.~\ref{fig:therm-iso-channels}\subref{fig:thermal-channel} yields an isometry $V^{\eta,\bar{n}_B}_{A\to BWE}$ shown in Fig.~\ref{fig:therm-iso-channels}\subref{fig:iso-channel}.
Thus, $\hat{\rho}_{B^n} \equiv\tr_{W^nE^n}\left( V^{(\eta,\nb)\otimes n}_{A\to BWE}(\hat{\rho}_{A^n})\right)$ and $\hat{\rho}_{W^n} \equiv\tr_{B^nE^n}\left( V^{(\eta,\nb)\otimes n}_{A\to BWE}(\hat{\rho}_{A^n})\right)$. We use $V^{\eta,\bar{n}_B}_{A\to BWE}$ in the proof of Theorem~\ref{thm:covert-capacity}.

\begin{figure*}[htb]
\centering
\includegraphics[width=0.9\textwidth]{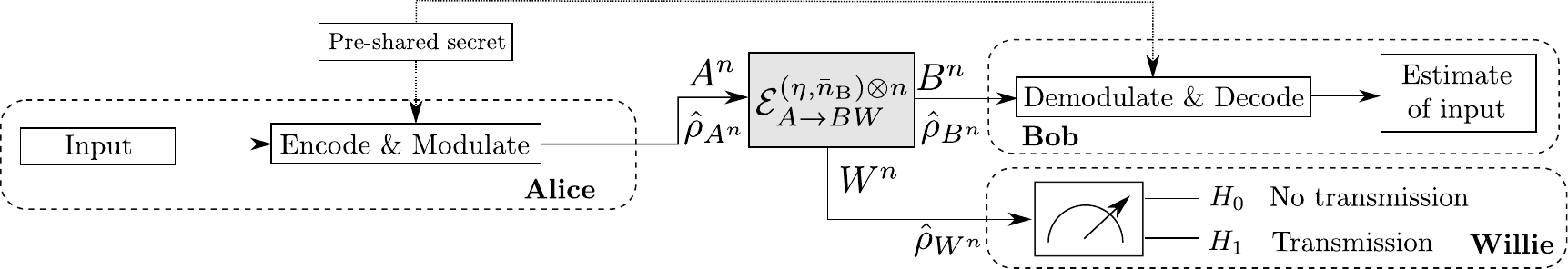}
\caption{System model for covert secrecy and covert entanglement generation. Alice either has an input and transmits, or she is quiet. When she has an input, she encodes and modulates a state $\hat{\rho}_{A^n}$ of system $A^n$ before transmitting it over $n$ uses of the bosonic channel $\mathcal{E}^{\eta,\nb}_{A\to BW}$ depicted in Fig.~\ref{fig:therm-iso-channels}\protect\subref{fig:thermal-channel}. Bob receives state $\hat{\rho}_{B^n}$, demodulates and decodes to estimate the input. Warden Willie receives state $\hat{\rho}_{W^n}$ and uses it to decide between hypotheses $H_0$ and $H_1$ corresponding to a quiet or transmitting Alice.
For covert secrecy, Alice uses position-based coding with a pre-shared secret key $k$, unknown to Willie, to encode the message $m$ into a coherent-state codeword modulated using quadrature phase-shift keying (QPSK). Bob employs sequential decoding with the pre-shared secret $k$ to estimate $m$.
For entanglement generation, Alice prepares a maximally-entangled state $\ket{\Phi}\bra{\Phi}_{RM}$. She encodes the state of system $M$ as a superposition of the codewords from the aforementioned secrecy codebook in system $A^n$. Bob constructs a coherent version of the sequential decoding scheme used in the secrecy construction to recover a state entangled with the reference system $R$.
}
\label{fig:system}
\end{figure*}
Now consider the covert communication setting described in Fig.~\ref{fig:system}.
Alice wishes to transmit a quantum state $\hat{\rho}_{A^n}$ to Bob over $n$ uses of the bosonic channel $\mathcal{E}_{A\to BW}^{\eta,\bar{n}_\textrm{B}}$ without being detected by Willie. For covert secret communication, Alice encodes a message $m$ into the $A^n$ subsystems via classical modulation, aiming to keep it covert and secret from Willie while allowing Bob to decode it reliably. Covertness, secrecy, and reliability are respectively defined in Sections \ref{sec:hypothesis-testing}, \ref{sec:secrecy}, and \ref{sec:reliability}. For covert entanglement generation, she instead creates local entanglement and maps half of the entangled state to the $A^n$ subsystems. Alice seeks to ensure Bob can reliably obtain half of an entangled state, where the other half is maintained locally by her. 

Finally, we assume that Alice and Bob pre-share a classical secret, as is standard in covert communications \cite{bash12sqrtlawisit, bash13squarerootjsac,bloch15covert, wang15covert, bash15covertbosoniccomm, bullock20discretemod, gagatsos20codingcovcomm, azadeh16quantumcovert-isitarxiv, bullock2025fundamentallimitscovertcommunication, bullockCovertCommunicationClassicalQuantum2023, zlotnick25eacqcovert, bash15covertcommmag, chen23covcommssurvey}.
This aligns with the ``best practices'' of secure system design, as the security relies solely on the shared secret \cite{kerckhoffs1883law, talb2006}. Although this may seem restrictive, in many practical settings the cost of detection can far outweigh that of distributing a secret. Furthermore, just as in the classical regime \cite{che13sqrtlawbscisit, bloch15covert}, a pre-shared key may be unnecessary when Alice’s channel to Bob is better than her channel to Willie. However, in practice, achieving such a channel advantage can be more challenging than pre-sharing a secret.

\subsection{Hypothesis Testing and Covertness}\label{sec:hypothesis-testing}
We assume that Willie has complete knowledge of the system in Fig.~\ref{fig:system}, except for Alice and Bob's pre-shared secret.
Willie must determine, from his channel output, whether Alice is transmitting (hypothesis $H_1$) or not (hypothesis $H_0$).

Let $\hat{\rho}_{W^n}^{(0)}\equiv\tr_{B}\left(\mathcal{E}_{A\to BW}^{(\eta,\bar{n}_\textrm{B})\otimes n}(\ket{0}\bra{0}_{A}^{\otimes n}
)\right)$ be the state Willie observes when Alice is quiet, where $\ket{0}\bra{0}_A$ is a vacuum state and is the ``innocent'' input. Similarly, we denote $\hat{\rho}_{W^n}$ as the state Willie receives when Alice transmits. When Alice is quiet, she inputs vacuum, so $\hat{\rho}^{(0)}_{W^n}=\left(\hat{\rho}^{(0)}_W\right)^{\otimes n}$ \, where $\hat{\rho}^{(0)}_W = \hat{\rho}_{\textrm{th}}(\eta\nb)$ is an attenuated thermal state \cite{weedbrook12gaussianQIrmp}.

Willie has to determine if Alice transmits. Thus, over $n$ channel uses, he attempts to distinguish between $\hat{\rho}_{W^n}$ and $\hat{\rho}^{(0)}_{W^n}$. The null and alternate hypotheses $H_0$ and $H_1$ correspond to Alice being quiet and transmitting, respectively. 
Willie is assumed to have access to arbitrary quantum resources to discriminate between $H_0$ and $H_1$, including fault-tolerant quantum computers, perfect quantum measurement, and ideal quantum memories. We also allow Willie to collect any photons that do not reach Bob.

For equal priors, i.e., $P(H_1)=P(H_0)=\frac{1}{2}$, Willie's  error probability is 
$P^{\rm (e)}_W=\frac{P_\textrm{FA}+P_\textrm{MD}}{2}$,
with probabilities of false alarm and missed detection $P_{\textrm{FA}}=P(\text{choose } H_1|H_0 \text{ true})$ and $P_\textrm{MD}=P(\text{choose } H_0|H_1 \text{ true})$, respectively. Willie achieves the trivial upper bound $P^{\rm (e)}_W\leq \frac{1}{2}$  by a random decision. Alice and Bob desire $\min P^{\rm (e)}_W$ close to that of this ineffective device: they seek
    $\min P^{\rm (e)}_W\geq\frac{1}{2}-\delta,$
where $\delta>0$ quantifies the covertness level. It is bounded by the trace distance between the output states under each hypothesis \cite[Sec.~9.1.4]{wilde16quantumit2ed}:
    $\min P^{\rm (e)}_W\geq\frac{1}{2}-\frac{1}{4}\left\|\hat{\rho}_{W^n}-\hat{\rho}^{(0)}_{W^n}\right\|_1$, 
where $\|\hat{A}\|_1\equiv\tr\left(\sqrt{\hat{A}^\dagger\hat{A}}\right)$ is the trace norm of $\hat{A}$ and $\hat{A}^\dagger$ is its Hermitian transpose. 
The trace distance is often mathematically unwieldy. Conveniently, the quantum relative entropy (QRE) 
 $D\left(\hat{\rho}\middle\|\hat{\sigma}\right) \equiv \tr\left(\hat{\rho}\log\hat{\rho} - \hat{\rho}\log\hat{\sigma}\right)$ is additive over product states, and upper bounds the trace distance via the quantum Pinsker's inequality \cite[Thm.~11.9.1]{wilde16quantumit2ed}: 
$\frac{1}{4}\left\|\hat{\rho}_{W^n}-\hat{\rho}^{(0)}_{W^n}\right\|_1 \leq \sqrt{\frac{1}{8} D\left(\hat{\rho}_{W^n}\middle\|\hat{\rho}^{(0)}_{W^n}\right)}$. 
We use QRE as our covertness criterion, as is standard in both classical \cite{bloch15covert, wang15covert} and quantum \cite{bullock20discretemod, gagatsos20codingcovcomm, azadeh16quantumcovert-isitarxiv, bullock2025fundamentallimitscovertcommunication} analyses.
Formally, we call a communication scheme \emph{covert} if, for $\delta_{\rm C}>0$, $D\left(\hat{\rho}_{W^n}\middle\|\hat{\rho}^{(0)}_{W^n}\right) \leq \delta_{\rm C}.$

\subsection{Secrecy}
\label{sec:secrecy}
Covertness yields a probabilistic guarantee.
Thus, in the unlikely event that Alice’s transmission is detected, Willie could potentially decode the message unless secrecy requirements are also met.
That is, we desire to bound the amount of information Willie gains if he detects Alice's state.
Given a message $m\in\mathcal{M}$, Alice prepares a quantum state $\hat{\rho}^{(m)}_{A^n}$ and transmits it through $n$ uses of a quantum channel $\mathcal{N}_{A\to BW}$. Willie recovers $\hat{\rho}_{W^n}^{(m)} = \tr_{B^n}(\mathcal{N}^{\otimes n}_{A\to BW}(\hat{\rho}_{A^n}^{(m)}))$. We call a system \emph{secret} if there exists a constant ``non-informative'' state $\hat{\breve{\rho}}_{W^n}$ that does not depend on Alice's original message, such that the leakage distance is $\delta_{\rm S}$-small:
$\max_{m\in\mathcal{M}}\| \hat{\rho}_{W^n}^{(m)} - \hat{\breve{\rho}}_{W^n}\|_1 \le \delta_{\rm S}$.

\subsection{Decoding Reliability}
\label{sec:reliability}
For transmission of classical information, we must ensure that Bob can reliably decode. A coding scheme $\mathcal{C}$ is reliable if, for any $\epsilon_{\rm c}>0$, 
    $\bar{P}_{\rm e}(\mathcal{C})\leq \epsilon_{\rm c}$, 
where $\bar{P}_{\rm e}(\mathcal{C})$ denotes the average probability of error over the message and key. 
We call an entanglement generation scheme $\mathcal{G}$ reliable if, for any $\epsilon_{\rm g}>0$, 
    $F\left(\hat{\Phi}_{RM},\hat{\tau}^{(\mathcal{G})}_{R\check{M}}\right)\geq 1-\epsilon_{\rm g}$, 
where $\hat{\Phi}_{RM}$ is the initial maximally-entangled state at Alice, $\hat{\tau}^{(\mathcal{G})}_{R\check{M}}$ is the entangled state Bob recovers, and $F\left(\hat{\rho}, \hat{\sigma} \right) \equiv \left[\tr\left( \sqrt{\sqrt{\hat{\rho}} \hat{\sigma}\sqrt{\hat{\rho}}}\right)\right]^2$ is the quantum fidelity between states $\hat{\rho}$ and $\hat{\sigma}$.

\section{Fundamental Limits on Covert Secrecy and Entanglement Generation}\label{sec:primary-results}

\begin{table}[h]
\centering
\caption{List of Symbols used in Section~\ref{sec:primary-results}}
\begin{tabular}{lp{5.4cm}}
\hline
\textbf{Symbol} & \textbf{Description} \\
\hline
$n$       & Total number of channel uses \\
$R_{\rm sec}$ & Achievable covert secrecy rate \\
$L_{\rm sec}$ & Covert secrecy capacity \\
$A, B, W, E, R$ & Quantum systems containing Alice, Bob, Willie, the environment, and reference states \\
$X$ & Random variable for public shared randomness \\
$\mathcal{M}, \mathcal{K}, \mathcal{C}$ & Sets of the message, key, and code, respectively \\
$m, k, c$ & Instances of the message, key, and code \\
$\delta_{\rm{S}}, \delta_{\rm{C}}$   & Secrecy and covertness constraint parameters \\
$\epsilon_{\rm c}, \epsilon_{\rm g}$      & Error parameters defined for reliable secrecy and entanglement generation  \\
$\ns$ & Mean photon number of Alice's transmission\\
$\{\hat{\Omega}_{B^n}^{(m,k)}\}$ & Bob's sequential-decoding POVM\\
$\hat{\mathcal{D}}_{B^n \to B^n \check{M} \check{K}_1 \check{K}_2}$ & Coherent version of $\{\hat{\Omega}_{B^n}^{(m,k)}\}$ \\ 
$\eta$        & Transmittance of the channel \\
$\nb$    & Mean thermal photon number of the channel \\
$\mathcal{E}^{\eta,\nb}_{A\to BW}$ & Bosonic channel parameterized by $\eta$ and $\nb$\\
$V^{\eta,\nb}_{A\to BWE}$ & Isometric extension of $\mathcal{E}^{\eta,\nb}_{A\to BW}$ \\
$c_{\rm{cov}}, c_{\rm{rel}}, c_{\rm{key}}$  & Covertness, reliability, and key-length constants\\ 
$\varsigma_n^{(i)}$ & Arbitrary bounding variables, $\varsigma_n^{(i)}\in o(1)\cap\omega(1/\sqrt{n})$\\
$\zeta_n^{(i)}$ &  Arbitrary bounding variables, $\zeta_n^{(i)}\in \omega(1/\sqrt{n})$\\
$\mu_n$ & Arbitrary bounding variable, $\mu_n \in o(1)$ \\
$\hat{\rho}_{\textrm{th}}(\nb)$ & Thermal state with mean thermal photon number $\nb$ \\
$\hat{\rho}^{(0)}_{W^n}$ & Willie's state when Alice does not transmit, $\hat{\rho}_{\textrm{th}}$$(\eta\nb)^{\otimes n}$ \\
$\hat{\rho}^{}_{W^n}$ & Willie's state when Alice transmits\\
$\hat{\breve\rho}_{W}$ & Willie's state when Alice transmits a uniform mixture of QPSK-modulated coherent states\\
$P_{\rm e}^{(m,k)}(\mathcal{C})$ & Decoding error probability for a code $\mathcal{C}$ defined by $m,k$ inputs\\
\hline
\end{tabular}
\label{tab:symbols}
\end{table}

Here, we derive the fundamental limits of covert entanglement generation over the bosonic channel.
We adapt the approach from \cite{kimelfeld25ceg} by adding a secrecy requirement to the existing results \cite{bullock20discretemod, gagatsos20codingcovcomm, wang22isitcoverttd, wang22isittowardtdjournal} and then convert the resulting code to its coherent version \cite[Sec.~5.4]{wilde16quantumit2ed}, \cite{devetakPrivateClassicalCapacity2005,hayden2008decoupling}. 

To adapt the covert secrecy proof to an infinite-dimensional channel, we first show that we may obtain secrecy via a sufficiently large position-based codebook \cite[Lem.~V.1]{wang22isittowardtdjournal}, and then use continuity of entropy under finite energy constraints \cite{winter16cont} to show that this also guarantees covertness. Importantly, the resulting non-informative state $\hat{\breve{\rho}}_W^{\otimes n}$ is also a quantum-secure covert state: it satisfies the covertness requirement $D\left(\hat{\breve{\rho}}_W^{\otimes n}\|\hat{\rho}_{W^n}^{(0)}\right)\leq\delta_{\rm C}$ from Section \ref{sec:hypothesis-testing} when input mean photon number per mode $\ns$ is controlled per \cite[Th.~2]{bullock20discretemod}.

However, adapting the covert entanglement generation proof to the bosonic channel poses different challenges. Since the bosonic channel is not isometric, we must use the channel's Stinespring dilation (Fig.~\ref{fig:therm-iso-channels}) and show that we can fully decouple our output state from both Willie and the environment. We obtain decoupling by employing a one-time pad, but this alone does not guarantee covertness if Willie's channel is superior to Bob's, since the resulting non-informative state is not a quantum-secure covert state. Thus, we ensure covertness by using an additional key, following established methods \cite[Th.~1]{gagatsos20codingcovcomm}, \cite[Lem.~V.1]{wang22isitcoverttd}.

For convenience, Table \ref{tab:symbols} lists commonly used symbols in this section and their descriptions.

\subsection{Covert and Secret Communication}\label{subsec:cov-sec-comms}

Systems that achieve the secrecy requirement in Section \ref{sec:secrecy} without the covertness guarantee of Section \ref{sec:hypothesis-testing} have rates in bits per channel use, expressed as $\frac{\log |\mathcal{M}|}{n}$ for a message set $\mathcal{M}$ and $n$ channel uses. In the covert setting, due to the SRL, $\log|\mathcal{M}|\in \mathcal{O}(\sqrt{n})$, yielding a zero-capacity result as $n\to\infty$. Hence, we define a covert secrecy rate as $R_{\rm sec} = \frac{\log |\mathcal{M}|}{\sqrt{n\delta_{\rm C}}}$.

\begin{definition}
(Achievable covert secrecy rate). A covert secrecy rate $R_{\rm sec}$ is \emph{achievable}, if for large enough number of channel uses $n$, and for every $\epsilon_{\rm c} \in (0,1), \delta_{\rm S} >0, \delta_{\rm C}>0$, there exists a $(e^{R_{\rm sec}\sqrt{n\delta_{\rm C}}}, n, \epsilon_{\rm c}, \delta_{\rm S}, \delta_{\rm C})$-code for secret and covert classical communication. 
\end{definition}

\begin{definition}\label{def:covseccap}
(Covert secrecy capacity). The covert secrecy capacity $L_{\rm sec}$ is the supremum over all achievable covert secret rates.
\end{definition}

\begin{theorem}\label{thm:covert-capacity-secrecy}
    The covert secrecy capacity of a bosonic channel $\mathcal{E}^{(\eta,\nb)}$ 
    is $L_{\rm sec}\left(\mathcal{E}^{(\eta,\nb)}\right) = c_{\rm cov}c_{\rm rel}$,
    where \begin{align}
        c_{\rm cov}&=\frac{\sqrt{2\eta\nb(1+\eta\nb)}}{1-\eta} \label{eq:ccov}\\
    c_{\rm rel}&=\eta\log\left(1+\frac{1}{(1-\eta)\nb}\right). \label{eq:crel}
    \end{align}
\end{theorem}
\begin{IEEEproof}
\noindent{\bf Achievability:}\\
\noindent{\bf Construction:} Alice employs position based coding \cite{wilde17ppc} by generating public shared randomness $\left(X^n\right)^{\otimes |\mathcal{M}||\mathcal{K}|}$ and distributes copies to Bob and Willie, where $X$ is a uniform random variable over $\{1,\ldots,4\}$ and $X^n$ is an independent and identically-distributed (i.i.d.) random vector. Alice and Bob pre-share a key $k\in\mathcal{K}$ that is kept secret from Willie. Based on message $m$, Alice subselects codeword $X^n(m,k)$ indexed by $(m,k)$ and encodes it as a product of quadrature phase-shift keyed (QPSK) coherent states: $\ket{\phi\left(X^n(m,k)\right)}_{A^n}=\ket{\sqrt{\ns}e^{j\pi X_1/2 }}_{A_1}\otimes\cdots\otimes\ket{\sqrt{\ns}e^{j\pi X_n/2 }}_{A_n}$ where $\ns$ is the mean photon number per mode and $\pi X_i / 2$, $i=1,\ldots,n$, is the phase. She transmits the state over $n$ bosonic channel uses. Based on his knowledge of $k$, Bob uses the sequential-decoding positive operator-valued measurement (POVM) $\{\hat{\Omega}_{B^n}^{(m,k)}\}$ on his output systems $B^n$ \cite{wilde17ppc}. The following lemma yields that this scheme is reliable, covert, and secret, on average, for an appropriate choice of $(\mathcal{M},\mathcal{K})$:
\begin{lemma}\label{lem:secrecyach}
Consider the bosonic channel $\mathcal{E}^{\eta,\nb}_{A\to BW}$ described in Fig.~\ref{fig:therm-iso-channels}\subref{fig:thermal-channel} and let $\hat{\breve\rho}_{W}^{\otimes n}=\tr_{B^n}\left[\mathcal{E}^{\eta,\nb}_{A\to BW}((\hat{\rho}_A)^{\otimes n})\right]$ with $\hat{\rho}_A=\mathrm{tr}_X[\hat{\rho}_{XA}]$, where $\hat{\rho}_{XA} = \frac{1}{4}\sum_{x=1}^4 \ketbra{x}{x}_X \otimes \ketbra{\sqrt{\ns}e^{j\pi x/2}}{\sqrt{\ns}e^{j\pi x/2}}_A$. 
There exists a random coding scheme defined on QPSK coherent-state codeword input states with mean photon number per mode $\ns$ and $\varsigma_n\in o(1)\cap\omega(1/\sqrt{n})$ such that, for $\ns\in o(1)$ and for $n$ large enough, 
\begin{align}
    \log |\mathcal{M}| &= (1-\varsigma_n) c_{\rm rel}\ns n\\
    \log |\mathcal{K}| &= (1+\varsigma_n)c_{\rm key}\ns n\label{eq:keysizerandom}
\end{align}
where $c_{\rm rel}$ is defined in \eqref{eq:crel} and 
\begin{align}
    c_{\rm key}=( 1-\eta )\log\left(1+\frac{1}{\eta\nb}\right),\label{eq:cres}
\end{align}
while, letting $E_\mathcal{C}[\cdot]$ denote expectation over codebook $\mathcal{C}$, 
\begin{align}
    E_\mathcal{C}\left[\bar{P}_{\rm e}(\mathcal{C})\right] &\leq e^{-\varsigma_n^{(1)}\sqrt{n}} \label{eq:relcode} \\
    E_\mathcal{C}\left[\left|D(\hat{\bar\rho}_{W^n}\|\hat{\breve\rho}_W^{\otimes n})-D(\hat{\breve\rho}_W^{\otimes n}\|\hat{\rho}_{W^n}^{(0)})\right|\right] &\leq e^{-\varsigma_n^{(2)}\sqrt{n}}\label{eq:covertnesscode}\\
    \max_m E_\mathcal{C} \left[\left\|\hat{\bar{\rho}}^m_{W^n}-\hat{\breve\rho}_{W}^{\otimes n}\right\|_1\right]&\leq e^{-\varsigma_n^{(3)}\sqrt{n}}\label{eq:seccode}
\end{align}
 where, for each $i$, $\varsigma_n^{(i)}\in o(1)\cap\omega(1/\sqrt{n})$ and $\|\cdot\|_1$ is defined in Section \ref{sec:hypothesis-testing}.
\end{lemma}
To prove Lemma \ref{lem:secrecyach}, we adapt the position-based coding and sequential decoding strategy from \cite[Lem.~V.1]{wang22isitcoverttd},\cite{wilde17ppc} to ensure secrecy at the cost of a larger key requirement (Lemma \ref{lem:PBCSQ} in Appendix \ref{ap:oneshotproof}). 
We then show that our scheme is covert by setting the mean signal photon number per channel use as $\ns=c_{\rm cov}\sqrt{\frac{\delta}{n}}$ according to \cite{bullock20discretemod}.
The full proof of Lemma \ref{lem:secrecyach} is in Appendix \ref{ap:cov-sec-bosonic}.

Next, we show that a deterministic coding scheme exists that satisfies the average decoding reliability, covertness, and message average secrecy requirements. We then use the standard expurgation argument to construct bounds on the maximum probability of decoding error and achieve semantic secrecy. 
\begin{lemma}\cite[Lem.~10]{kimelfeld25ceg}\label{lem:expurge}
    There exists a sequence of deterministic coding schemes $\mathcal{C}=\{x^n(m,k)\}$ such that, for $\varsigma_n\in o(1)$ and $n$ large enough, 
\begin{align}
    \log |\mathcal{M}| &= (1-\varsigma_n) c_{\rm rel}\ns n\\
    \log |\mathcal{K}| &= (1+\varsigma_n) c_{\rm key}\ns n\label{eq:keysizedet}
\end{align}
    where 
    \begin{align}
      \max_{m,k} P_{\rm e}^{(m,k)}(\mathcal{C}) &\leq e^{-\zeta_n^{(1)}\sqrt{n}}\label{eq:detrelexp} \\
       \left|D(\hat{\bar\rho}_{W^n}\|\hat{\breve\rho}_W^{\otimes n})-D(\hat{\breve\rho}_W^{\otimes n}\|\hat{\rho}_{W^n}^{(0)})\right| \label{eq:detcovexp}&\leq e^{-\zeta_n^{(2)}\sqrt{n}}\\
    \max_m\left\|\hat{\bar{\rho}}^m_{W^n}-\hat{\breve\rho}_{W}^{\otimes n}\right\|_1&\leq e^{-\zeta_n^{(3)}\sqrt{n}} \label{eq:detsecexp}
\end{align}
with $\zeta_n^{(i)} \in \omega(1/\sqrt{n})$ for $i=1,2,3$.
\end{lemma}
The proof of Lemma \ref{lem:expurge} is in Appendix \ref{ap:lemderandom}.

Now, we show that the message rate achieved for the deterministic coding scheme in Lemma \ref{lem:expurge} converges to the capacity while maintaining covertness:
 $\log |\mathcal{M}| \geq (1-\varsigma_n)\eta\log\left(1+\frac{1}{(1-\eta)\bar{n}_b}\right)\ns n$. Note that \eqref{eq:detrelexp} implies decoding reliability, \eqref{eq:detsecexp} implies secrecy, and \eqref{eq:detcovexp} implies     $D\left(\hat{\bar \rho}_{W^n}\|\hat{\rho}_{W^n}^{(0)}\right)\leq n D\left(\hat{\breve \rho}_{W}\|\hat{\rho}_W^{(0)}\right) + e^{-\zeta_n^{(2)}\sqrt{n}}$, 
  where $\hat{\breve \rho}_{W}$ is Willie's output from single-mode QPSK coherent state constellation input with mean photon number per mode $\ns$ and $\zeta_n^{(2)}\in\omega(1/\sqrt{n})$.
 Thus, choosing $\ns=c_{\rm cov}{\sqrt{\frac{\delta_{\rm c}}{n}}}$ and \cite[Thm.~2]{bullock2025fundamentallimitscovertcommunication} implies the scheme is covert as $n\to\infty$. Now, using Definition \ref{def:covseccap}, we have 
 \begin{align}
      L_{\rm sec}\left(\mathcal{E}^{(\eta,\nb)}\right)& \geq\lim_{n\to\infty}\frac{\log |\mathcal{M}|}{\sqrt{\delta_{\rm C} n}} = c_{\rm cov}c_{\rm rel}.
 \end{align} 
\noindent{\bf Converse:}
The converse matches that for the non-secret covert case in \cite[Thm.~1]{gagatsos20codingcovcomm}, as an upper bound on covert rate without secrecy is  also an upper bound on this rate with secrecy. 
\end{IEEEproof}

\subsection{Covert Entanglement Generation}\label{subsec:cov-eg}
Similar to the covert secrecy rate, the entanglement generation rate is defined as $\frac{\log(\dim(\mathcal{H_M}))}{n}$ in ebits generated per channel use, where $\dim(\mathcal{H_M})$ is the dimension of the entangled state. However, the SRL also governs entanglement generation, and we define its covert rate as $R_{\rm EG} = \frac{\log(\dim(\mathcal{H_M}))}{\sqrt{n\delta_{\rm C}}}$.

\begin{definition}
(Achievable covert entanglement-generation rate). A covert entanglement-generation rate is \emph{achievable}, if for large enough uses of the channel $n$, and for every $\epsilon_{\rm g} \in (0,1), \delta_{\rm C}>0$, there exists a $(e^{R_{\text{EG}}\sqrt{n\delta_{\rm C}}}, n, \epsilon_{\rm g}, \delta_{\rm C})$-code for covert entanglement-generation. 
\end{definition}

\begin{definition}
(Covert entanglement-generation capacity). The covert entanglement generation capacity $L_\text{EG}$ is the supremum over all \emph{achievable} \emph{covert} entanglement generation rates.
\end{definition}

\begin{theorem}\label{thm:covert-capacity}
    The covert entanglement-generation capacity of a bosonic channel $\mathcal{E}^{(\eta,\nb)}$ 
    is $L_\text{EG}\left(\mathcal{E}^{(\eta,\nb)}\right) =  c_\textrm{cov} c_{\rm rel}$, 
    where $c_{\rm cov}$ and $c_{\rm rel}$ are defined in \eqref{eq:ccov} and \eqref{eq:crel}, respectively.
\end{theorem}
\begin{remark}
The achievability proof adapts \cite[Thm.~2]{kimelfeld25ceg} by modifying the classical code construction in the proof of Theorem \ref{thm:covert-capacity-secrecy}, and converting it to a quantum one. The decoupling argument in \cite{kimelfeld25ceg} relies on the channel from Alice to Bob and Willie being an isometry. Since the bosonic channel model is not an isometry from Alice to Bob and Willie, we consider its Stinespring dilation depicted in Fig.~\ref{fig:therm-iso-channels}\subref{fig:iso-channel}. To generate entanglement effectively, Bob must decouple his state from Willie's and the environment's systems $W^nE^n$. We, therefore, modify the code construction of Theorem \ref{thm:covert-capacity-secrecy} as follows: 
\end{remark}
\begin{IEEEproof}
\noindent{\bf Achievability:}\\
\noindent{\bf Modified secret code construction:}
Alice and Bob use the position-based coding and sequential decoding strategy described in the proof of Theorem \ref{thm:covert-capacity-secrecy}, with public shared classical resource $(X^n)^{\otimes |\mathcal{M}||\mathcal{K}_1|}$. Alice randomly selects and pre-shares a secret key $k_1\in \mathcal{K}_1$ with Bob, enabling covertness.
Further, Alice randomly selects and pre-shares a second secret key $k_2 \in \mathcal{K}_2$ with Bob to be used as a one-time pad, where $\log|\mathcal{K}_2|=\log|\mathcal{M}|$. 
Then, Alice takes her message-key tuple and encodes into a product of QPSK-modulated coherent-states $\ket{\phi(X^n(m\oplus k_2, k_1))}\triangleq\ket{x^n_{\rm coh}(m\oplus k_2,k_1)}_{A^n}=\ket{\sqrt{\ns}e^{j\pi x_1/2}}_{A_1}\otimes\cdots\otimes\ket{\sqrt{\ns}e^{j\pi x_n/2}}_{A_n}$ according to the position-based codebook, and where $m\oplus k_2$ denotes application of $k_2$ as a one-time pad. She then transmits over the bosonic channel. Bob, using $k_1$, decodes employing the position-based decoding strategy, then recovers an estimate of the message using $k_2$ to undo the one-time pad. The following lemma proves that this strategy is reliable, covert against an adversary with access to systems $W^n$, and secret against the joint systems $W^nE^n$.  
\begin{lemma}\label{lem:detmodsec}
    There exists a sequence of deterministic coding schemes $\mathcal{C}$ such that, for $\varsigma_n\in o(1)$ and $n$ large enough,
    \begin{align}
        \log |\mathcal{M}| &= (1-\varsigma_n) c_{\rm rel} \ns n\\
        \log |\mathcal{K}_1| &= \left[(1+\varsigma_n) c_{\rm key} -(1-\varsigma_n) c_{\rm rel}\right]^+ \ns n\\
        \log|\mathcal{K}_2|&=(1-\varsigma_n)c_{\rm rel}\ns n,
    \end{align}
    where $[a]^+\equiv\max(a,0)$, while
    \begin{align}
              \max_{m,k} P_{\rm e}^{(m,k)}(\mathcal{C}) &\leq e^{-\zeta_n^{(1)}\sqrt{n}}\label{eq:detmodrelexp}, \\
       \left|D(\hat{\bar\rho}_{W^n}\|\hat{\breve\rho}_W^{\otimes n})-D(\hat{\breve\rho}_W^{\otimes n}\|\hat{\rho}_{W^n}^{(0)})\right| \label{eq:detmodcovexp}&\leq e^{-\zeta_n^{(2)}\sqrt{n}},\\
  \text{ for every $m$,  }   \hat{\bar{\rho}}^m_{E^nW^n}&=\hat{\bar{\rho}}_{E^nW^n}, \label{eq:detmodsecexp}
    \end{align}
\end{lemma}
\noindent with $\zeta_n^{(i)} \in \omega(1/\sqrt{n})$ for $i=1,2$.
\noindent Here, \eqref{eq:detmodrelexp} follows from \cite[Sec.~5.1]{wang22isittowardtdjournal}, \eqref{eq:detmodcovexp} follows from \cite[Eq. (8)]{wang22isittowardtdjournal} combined with \eqref{eq:covbnd1}-\eqref{eq:samepower}, and \eqref{eq:detmodsecexp} is true by construction.\\
\noindent{\bf Code conversion}: As in \cite[Thm.~2]{kimelfeld25ceg}, we convert the classical code from the proof of Theorem \ref{thm:covert-capacity-secrecy} into an entanglement generation code.
Consider the classical QPSK codebook $\{x^n(m\oplus k_1,k_2)\}_{m,k_1,k_2}$ from the modified secret code construction. Alice converts this to a quantum codebook $\{\ket{\phi_m}_{A^n} : m\in\mathcal{M}\}$ with $\ket{\phi_m}_{A^n}=\frac{1}{\sqrt{|\mathcal{K}_1||\mathcal{K}_2|}}\sum_{k_1,k_2} e^{j f(m\oplus k_2,k_1)}\ket{x^n_{\rm coh}(m\oplus k_1,k_2)}_{A^n}$,
where $f(\cdot,\cdot)$ is defined later. Alice prepares her encoding by generating a maximally entangled state: $\ket{\Phi}_{RM}=\frac{1}{\sqrt{|\mathcal{M}|}}\sum_{m}\ket{m}_R\otimes\ket{m}_M$,
where subsystems $R$ and $M$ are the resource and message, respectively. She generates a copy of the $M$ subsystem using a CNOT gate as in \cite[Eq.~(108)]{kimelfeld25ceg} to obtain $\ket{\tau}_{RMM^\prime}$. 
Alice applies an isometry $\hat{U}_{M^\prime\to A^n}$ that takes $\ket{m}_{M^\prime}\to\ket{\phi_m}_{A^n}$ as 
\begin{align}
        \ket{\tau}_{RMA^n}&= \left(\hat{I}\otimes\hat{I}\otimes\hat{U}_{M^\prime\to A^n}\right)\ket{\tau}_{RMM^\prime}\\
    &=
    \frac{1}{\sqrt{|\mathcal{M}|}}
    \sum_{m}\ket{m}_R\otimes\ket{m}_M \otimes\ket{\phi_m}_{A^n}. \label{eq:alice-quantum-encoding}
\end{align}
She transmits $A^n$ systems over $n$ copies of $V^{\eta,\bar{n}_b}_{A\to BWE}$. We represent the global state as $\ket{\tau}_{RMB^nW^nE^n} = \frac{1}{\sqrt{|\mathcal{M}|}}\sum_{m}\ket{m}_R\otimes\ket{m}_M \otimes\ket{\phi_m}_{B^nW^nE^n}$
with $\ket{\phi_m}_{B^n W^nE^n}=V^{(\eta,\bar{n}_b)\otimes n}_{A\to BWE}(\ket{\phi_m}_{A^n}) $ being the channel output given input $\ket{\phi_m}_{A^n}$ and $V^{\eta,\nb}_{A\to BWE}$ being the Stinespring representation of the bosonic channel $\mathcal{E}^{\eta,\bar{n}_b}_{A\to BW}$.
Recall that Lemma \ref{lem:detmodsec} implies that there is a decoding POVM $\{\hat{\Omega}_{B^n}^{(m,k_1,k_2)}\}$ such that 
the classical code achieves $ \tr\left[\hat{\Omega}_{B^n}^{(m,k_1,k_2)}\hat{\rho}_{B^n}^{(m,k_1,k_2)}\right] \geq 1-e^{-\zeta_n^{(1)}\sqrt{n}}$, for all $m,k_1,k_2$.
We construct a coherent version of this POVM \cite[Sec.~5.4]{wilde16quantumit2ed} as $\hat{\mathcal{D}}_{B^n \to B^n \check{M} \check{K}_1 \check{K}_2} = \sum_{m,k_1,k_2} \sqrt{\hat{\Omega}_{B^n}^{(m,k_1,k_2)}} \otimes \ket{m}_{\check{M}} \otimes \ket{k_1}_{\check{K}_1}\otimes \ket{k_2}_{\check{K}_2}$ 
which, after its use, yields the global state:
\begin{align}
    &\ket{\tau}_{RMB^nW^nE^n\check{M}\check{K}_1\check{K}_2}\nonumber\\&=\left(\hat{I}_{RM}\otimes  \hat{\mathcal{D}}_{B^n\to B^n\check{M}\check{K}_1\check{K}_2}\otimes\hat{I}_{W^nE^n}\right)\ket{\tau}_{RMB^nW^nE^n}\label{eq:globaldecoded}
\end{align}

We now show that this conversion yields a reliable and covert entanglement-generation scheme.
\begin{lemma}\label{lem:entanglementgeneration}
    Consider covert entanglement generation via lossy thermal noise bosonic channel $\mathcal{E}^{\eta,\bar{n}_b}_{A\to BW}$ with Stinespring dilation  $V^{\eta,\bar{n}_b}_{A\to BWE}$. For any $\varsigma_n\in o(1)\cap\omega(1/\sqrt{n})$ there exists $\varsigma_n^{(1)},\varsigma_n^{(4)}\in o(1)\cap\omega(1/\sqrt{n})$ such that for $n$ sufficiently large
    \begin{align}
        \log d_M &\geq(1-\varsigma_n) c_{\rm rel}\ns \label{eq:covegenrate}n\\
        \log |\mathcal{K}_1| &= \left[(1+\varsigma_n) c_{\rm key} -(1-\varsigma_n) c_{\rm rel}\right]^+\ns n\\
        \log|\mathcal{K}_2|&=(1-\varsigma_n)c_{\rm rel}\ns n,
    \end{align}
    while 
    \begin{align}
        F(\hat{\Phi}_{RM},\hat{\tau}_{R \check{M}})&\geq 1-e^{-\varsigma_n^{(1)}\sqrt{n}}\label{eq:ceqrel}\\
\left|D\left(\hat{\tilde\rho}_{W^n}\middle\|\hat{\rho}_{W^n}^{(0)}\right)-D\left( \hat{\breve{\rho}}_{W^n}^{\otimes n}\middle\|\hat{\rho}_{W^n}^{(0)}\right)\right|&\leq e^{-\varsigma_n^{(4)}\sqrt{n}} \label{eq:ceqcov}   \end{align}
where $d_M=\dim(\mathcal{H}_M)$, $\hat{\Phi}_{RM}$ is the maximally entangled state, $\hat{\tau}_{R \check{M}}$ is Bob's decoded state and $\hat{\tilde\rho}_{W^n}$ is Willie's received state from the covert entanglement generation scheme. 
\end{lemma}
The proof of Lemma \ref{lem:entanglementgeneration} is adapted from \cite{kimelfeld25ceg}, and key steps are provided in Appendix \ref{ap:entanglementgeneration}. The main challenge in adapting \cite{kimelfeld25ceg} to infinite dimensions is to ensure that state approximation holds for non-orthogonal coherent state codewords and to carefully decouple Bob's state from both Willie's subsystems $W^n$ and the ancillary subsystems $E^n$. Further, showing that covertness is maintained requires the use of continuity of entropy results for bosonic systems \cite{winter16cont}.

We now show that this scheme with rate given by \eqref{eq:covegenrate} achieves the capacity in the limit $n\to\infty$ while remaining covert. Note that \eqref{eq:ceqrel} ensures reliability and \eqref{eq:ceqcov} implies 
     $D\left(\hat{\tilde \rho}_{W^n}\|\hat{\rho}_{W^n}^{(0)}\right)\leq n D\left(\hat{\breve \rho}_{W}\|\hat{\rho}_W^{(0)}\right) + e^{-\varsigma_n^{(4)}\sqrt{n}}$,
 where $\varsigma_n^{(4)}\in\omega(1/\sqrt{n})$.
 Thus, choosing $\ns=c_{\rm cov}{\sqrt{\frac{\delta}{n}}}$ implies the scheme is covert in the limit $n\to\infty$. Now, using Definition \ref{def:covseccap}, we have $L_\text{EG}\left(\mathcal{E}^{\eta,\nb}\right) \geq\lim_{n\to\infty}\frac{\log (\dim(\mathcal{H}_{\mathcal{M}}))}{\sqrt{\delta_{\rm C} n}} = c_{\rm cov}c_{\rm rel}$.

\noindent{\bf Converse:}
Here, we show that the decoding reliability and covertness conditions imply an upper bound on our entanglement generation rate. Specifically, if for any $\epsilon_{\rm g},\delta_{\rm C}>0$,
\begin{align}
    F\left(\hat{\Phi}_{RM},\hat{\tau}_{R\check{M}}^{(\mathcal{G})}\right)&\geq1-\epsilon_{\rm g}\label{eq:conv-rel} \text{ and }
    \\ D\left(\hat{\rho}_{W^n}\|\hat{\rho}_{W^n}^{(0)}\right)&\leq \delta_{\rm C},
\end{align}
then 
\begin{align}
    \log |\mathcal{M}|\leq (1+\mu_n)c_{\rm cov}c_{\rm rel}\sqrt{n\delta_{\rm C}}
\end{align}
for some $\mu_n\in o(1)$.

Note that \eqref{eq:conv-rel} implies:
\begin{align}
    \sqrt{\epsilon_{\rm g}}&\geq\sqrt{1-F\left(\hat{\Phi}_{RM},\hat{\tau}_{R\check{M}}^{(\mathcal{G})}\right)}\geq\frac{1}{2}\left\|\hat{\Phi}_{RM}-\hat{\tau}_{R\check{M}}^{(\mathcal{G})}\right\|_1\label{eq:conv-fuchs}\\
    &\geq\frac{1}{2}\left\|\frac{1}{|\mathcal{M}|}\hat{I}_M-\hat{\tau}_{\hat M}^{(\mathcal{G})}\right\|_1\label{eq:conv-trace-monotonicity}
\end{align}
where second inequality in \eqref{eq:conv-fuchs} is due to the Fuchs-van de Graaf inequalities \cite[Th.~1]{fuchsCryptographicDistinguishabilityMeasures1999} and \eqref{eq:conv-trace-monotonicity} is due to trace monotonicity. Note that \eqref{eq:conv-trace-monotonicity} is the classical reliability condition. Thus, the converse follows from the corresponding argument in the capacity proof for classical covert communication over the bosonic channel \cite[Th.~1]{gagatsos20codingcovcomm}.
\end{IEEEproof}

\section{Towards Practical Covert Entanglement Generation}
\label{sec:practical}

\begin{table}[h]
\centering
\caption{List of Symbols used in Section~\ref{sec:practical}}
\begin{tabular}{lp{6cm}}
\hline
\textbf{Symbol} & \textbf{Description} \\
\hline
$q$ & Probability Alice will transmit in a given channel use \\
$D_{\chi^2}\left(\hat{\rho}\middle\|\hat{\sigma}\right)$ & $\chi^2$-divergence for states $\hat{\rho}$ and $\hat{\sigma}$, $\operatorname{tr}[\hat{\rho}^2\hat{\sigma}^{-1}] - 1$\\
$\mathbf{x}$ & Pre-shared secret consisting of a binary sequence indicating selected channel uses for transmission\\
$\mathbf{y}$ & Pre-shared secret consisting of a quaternary sequence indicating choice of random Pauli operators \\
$w(\mathbf{x})$ & Weight, or total number of ones in $\mathbf{x}$ \\
$\vartheta$ & Arbitrary parameter required to remove expectation value of the number of transmitted ebits, $\vartheta>0$ \\
$p_I, p_X, p_Y, p_Z$ & Channel parameters from the combined Alice-to-Bob depolarizing and Pauli channel composition \\
$q_I, q_X, q_Y, q_Z$ & Alice-to-Bob Pauli channel parameters generated from Pauli twirling\\
$p_{\rm{f}}$ & Failure probability for Bob's projection operation \\
$\tau$ & Transmittance when a lossy thermal-noise channel is decomposed into a pure-loss + amplifier \\
$G$ & Gain parameter when a lossy thermal-noise channel is decomposed into a pure-loss + amplifier channel\\
$\mathcal{A}$ & A set of all $\mathbf{x}$ that is $\varepsilon$ close to $q$\\
$\hat{U}$ & Mappings defining the QECCs and covert QECCs\\
$\mathcal{Q}$ & Set of Pauli operators, $\{\hat{I}, \hat{X}, \hat{Y}, \hat{Z}\}$ \\
$\hat{P}^{(\mathbf{x},\mathbf{y})}$ & Sequence containing $\hat{I}$'s in innocent state positions in $\mathbf{x}$, and Pauli gates indicated in $\mathbf{y}$ otherwise \\
$\mathcal{P}^{\vec{p}}_{A\to B}$ & A Pauli channel parameterized by $\vec{p}$ \\
$R,M,\check{A},A$ & Systems containing reference state, Alice's entangled state, output state of the QECC, and Alice's input state into the channel  \\
$B, \check{B}, \check{M}$ & Systems containing the channel output state at Bob, Bob's projected state, and recovered entangled state \\ 
$\hat{\Pi}_B$ & Projection to the qubit basis from system $B$ to $\check{B}$ \\
$\ket{\Phi}_{RM}$ & Alice's entangled state with system $M$ sent to Bob\\
$\hat{\rho}_{RB^n}, \hat{\rho}_{W^n}$ & The states Bob and Willie receive when Alice uses a covert QECC over $n$ channel uses \\
$\hat{\tau}_{R\check{B}}$ & Bob's state after projection\\
$R_{\rm sr}, R_{\rm dr}$ & Achievable single- and dual-rail ebit rates\\
\hline
\end{tabular}
\label{tab:symbols-practical}
\end{table}

While the covert entanglement-generation capacity of the bosonic channel is achievable per Theorem~\ref{thm:covert-capacity}, it is unclear how to construct Alice's state in \eqref{eq:alice-quantum-encoding} physically. Hence, we investigate entanglement generation using single- and dual-rail photonic qubit modulation schemes in the lemmas that follow.

Let $[a]^+\equiv\max(a,0)$ and $H(\vec{p})=-\sum_{p_i\in\vec{p}} p_i\log_2(p_i)$ be the Shannon entropy associated with probability vector $\vec{p}$. A single-rail photonic qubit is encoded in a single mode of a photon represented by $\ket{\psi}=\alpha\ket{0}+\beta\ket{1}$, where $\ket{1}$ is a single-photon state. The dual-rail photonic qubit uses a single photon across two modes: $\ket{\psi}=\alpha\ket{01}+\beta\ket{10}$. Furthermore, $D_{\chi^2}\left(\hat{\rho}\middle\|\hat{\sigma}\right) = \operatorname{tr}[\hat{\rho}^2\hat{\sigma}^{-1}] - 1$ is the $\chi^2$-divergence between two states $\hat{\rho}$ and $\hat{\sigma}$.
Table \ref{tab:symbols-practical} lists commonly used symbols in this section and their descriptions.

\begin{figure*}[htb]
\centering
\includegraphics[width=\textwidth]{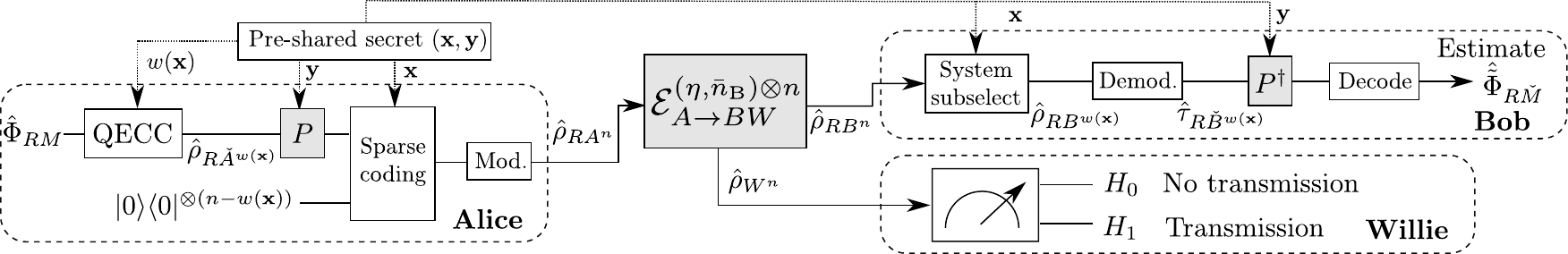}
\caption{Construction of achievable covert quantum entanglement generation over the bosonic channel $\mathcal{E}_{A\to BW}^{\eta,\bar{n}_\textrm{B}}$. Alice first prepares a Bell state $\hat{\Phi}_{RM}$. She keeps the reference subsystem $R$ and transmits the state of subsystem $M$ of $\hat{\Phi}_{RM}$ as follows: given a pre-shared secret $\mathbf{x}$, she applies a QECC corresponding to the number $w(\mathbf{x})$ of non-innocent output states, the output of the QECC is $\hat{\rho}_{R\check{A}^{w(\mathbf{x})}}$.
She then applies Pauli gates defined by a pre-shared secret sequence $\mathbf{y}$, followed by sparse coding that ``spreads'' these $w(\mathbf{x})$ non-innocent states across $n$ channel uses by inserting innocent input states according to locations in $\mathbf{x}$.
Alice then modulates and transmits the resulting state $\hat{\rho}_{RA^n}$. Bob receives $\hat{\rho}_{RB^n}$, sub-selects the systems containing the non-innocent states using $\mathbf{x}$, obtaining $\hat{\rho}_{RB^{w(\mathbf{x})}}$ and demodulates by projecting them onto the qubit basis, yielding $\hat{\rho}_{R\check{B}^{w(\mathbf{x})}}$. He then inverts the Pauli twirling operation using $\mathbf{y}$ and decodes to obtain a state $\hat{\tilde{\Phi}}_{R\check{M}}$ entangled with Alice's system $R$. Willie performs an optimal hypothesis test on whether transmission occurred.}
\label{fig:construction-sr-no-locc}
\end{figure*}

\begin{lemma}\label{lemma:sr-no-locc}
Using a single-rail photonic qubit, for sufficiently large $n$ and arbitrary $\vartheta>0$, at least $(1-\vartheta)\sqrt{n}\sqrt{2}c_\textrm{cov}R_{\rm sr}\sqrt{\delta_{\rm C}}$ ebits can be generated reliably and covertly in $n$ independent bosonic channel uses, where $c_\mathrm{cov}$ is in \eqref{eq:ccov}, $\delta_{\rm C}$ is the QRE-covertness constraint, and $R_{\rm sr}\geq\left[1-H(\vec{p})\right]^+$ is the constant achievable rate of reliable ebit generation per channel use for
    $\vec{p}=(p_I, p_X, p_Y, p_Z)$, $
    p_I = (1-\frac{3}{4}\pfail)q_I$, $
    p_i =(1-\frac{3}{4}\pfail)q_i+\frac{1}{4}\pfail$ for $i=X,Y,Z$ and
\begin{align}
    \pfail&=\frac{1+(1-\eta)\nb(3+2\nb-2\eta(\nb+\frac{1}{2}))}{(1+(1-\eta)\nb)^3} \\
    q_I&=\frac{1}{2N(G,\tau)}\left( \frac{G(2-\tau)-1+2\tau + 2\sqrt{G\tau}}{G^2} \right) \label{eq:qtw-I}\\
    q_X&= \frac{1}{2N(G,\tau)}\left( \frac{G(2-\tau)-1}{G^2} \right) \\
    q_Y&= \frac{1}{2N(G,\tau)}\left( \frac{G(2-\tau)-1}{G^2}\right) \\
    q_Z&= \frac{1}{2N(G,\tau)}\left( \frac{G(2-\tau)-1+2\tau - 2\sqrt{G\tau}}{G^2} \right), \label{eq:qtw-Z}
\end{align}
with $G=1+(1-\eta)\nb$, $\tau=\frac{\eta}{G}$, and $N(G, \tau) = \frac{G(2-\tau)+\tau-1}{2G^2}$.
\end{lemma}

\begin{IEEEproof}
\noindent{\bf Construction and reliability:} The construction is depicted in Fig.~\ref{fig:construction-sr-no-locc}. 
Alice prepares a Bell state $\ket{\Phi}_{RM}$. System $R$ is kept as reference, system $M$ is prepared to be sent through the channel. 
 To ensure covertness,
 Alice employs the sparse signaling approach from  \cite{tahmasbi21signalingcovert}:
 let $\mathbf{x}\equiv\{x_i, i=1,\ldots, n\}$ be a binary sequence indicating the selected channel uses for non-innocent transmission, with $x_i=1$ corresponding to a non-innocent qubit input from Alice on the $i^{\text{th}}$ channel use, and $x_i=0$ to innocent input. Thus, $w(\mathbf{x})\equiv\sum_{i=1}^{n}x_i$ is the number of non-innocent inputs for a given $\mathbf{x}$.
 For an arbitrary $\vartheta>0$, define $\mathcal{A}\equiv\left\{\mathbf{x}: \left|q-\frac{1}{n}w(\mathbf{x})\right|\leq\vartheta\right\}$
as the set containing length-$n$ binary sequences whose normalized weight is close to $q \in(0,1)$.
Let $ p_X(x)=\{q \text{~if~} x=1; 1-q \text{~if~} x=0\} \label{eq:p_X}.$
Denote $p(\mathcal{A})=\sum_{\mathbf{x}\in\mathcal{A}}\prod_{i=1}^{n}p_X(x_i)$, and: $    p_{\mathbf{X}}\left(\mathbf{x}\right)  \equiv \left\{\frac{\prod_{i=1}^{n}p_X(x_i)}{p(\mathcal{A})} \text{~if~} \mathbf{x}\in\mathcal{A}; 
         0 \text{~if~} \mathbf{x}\notin\mathcal{A}\right\}$.
Alice and Bob choose the channel uses for transmission by randomly sampling $\mathbf{x}\in\mathcal{A}$ using $p_{\mathbf{X}}$.
Their choice $\mathbf{x}$ comprises part of the classical pre-shared secret in Fig.~\ref{fig:construction-sr-no-locc}. 

Alice and Bob further generate a set $\{c_i\}_{i\in w(\mathcal{A})}$ of $|w(\mathcal{A})|= \lceil 2 \vartheta q n \rceil$ 
quantum error correction codes (QECCs) following standard techniques \cite[Sec.~24.6.3]{wilde16quantumit2ed}, where   $w(\mathcal{A})\equiv\{w(\mathbf{x}) : \mathbf{x}\in\mathcal{A}\}$.
For  $w(\mathbf{x})$ available non-innocent states,  the  QECC $c_{w(\mathbf{x})}$ yields an isometry $\hat{U}^{c_{w(\mathbf{x})}}$ that maps the input state $\hat{U}^{c_{w(\mathbf{x})}}\ket{\Phi}_{RM^{w(\mathbf{x})}}
=|c_{w(\mathbf{x})}(\Phi)\rangle_{R\check{A}^{w(\mathbf{x})}}$. We extend this to \emph{covert QECC} $c$ that maps $\hat{U}^{c(\mathbf{x})}:(\mathbf{x},|\Phi\rangle_{RM}) \rightarrow|c(\psi, \mathbf{x})\rangle_{RA^n}$, where $\hat{U}^{c(\mathbf{x})}=\left(\hat{I}_{A^{\mathbf{x}^c}}\otimes \hat{U}^{c_{w(\mathbf{x})}}\right)$ with $|c_{w(\mathbf{x})}(\psi)\rangle_{R\check{A}^{w(\mathbf{x})}}=\operatorname{tr}_{A^{\mathbf{x}^c}}\left(|c(\psi, \mathbf{x})\rangle_{A^n}\right)$ being the QECC for $w(\mathbf{x})$, and $\operatorname{tr}_{A_{\mathbf{x}^c}}$ the partial trace over the innocent inputs defined by $\mathbf{x}^c$, the one's compliment of $\mathbf{x}$. Given $\mathbf{x}$, the covert QECC can be thought of as the sparse encoding of the corresponding non-covert QECC of block length $w(\mathbf{x})$, with innocent states injected in each of the systems $\{A_i : x_i=0\}$.
Hence, the systems occupied by innocent states do not contribute to the code's error-correcting capabilities but are utilized for covertness.

Alice prepares $w(\mathbf{x})$ Bell state copies $\ket{\Phi}_{RM}^{\otimes w(\mathbf{x})}$ and applies the isometry $\hat{U}^{c_{w(\mathbf{x})}}$ to $M^{w(\mathbf{x})}$ subsystems. Alice and Bob use sequence $\hat{P}^{(\mathbf{x},\mathbf{y})}$ containing $\hat{I}$'s in innocent state positions in $\mathbf{x}$, and Pauli gates sampled uniformly at random from $\mathcal{Q}\equiv\{\hat{I},\hat{X},\hat{Y},\hat{Z}\}$ otherwise, with the choices being the pre-shared secret $\mathbf{y}$.  Alice applies $\hat{P}^{(\mathbf{x},\mathbf{y})}$ in the first stage of Pauli twirling \cite{emersonSymmetrizedCharacterizationNoisy2007}. Thus, Alice transmits $A^n$ subsystems of entangled encoded state $\ket{\Phi_c}_{R^{w(\mathbf{x})}A^n}=(\hat{I}_R^{\otimes w(\mathbf{x})}\otimes\hat{P}^{(\mathbf{x},\mathbf{y})}\hat{U}^{c(\mathbf{x})})\left(\ket{\mathbf{0}}_{A_{x^c}}\ket{\Phi}_{RM}^{\otimes w(\mathbf{x})}\right)$ given $\mathbf{x}, \mathbf{y}$,
over the channel $\mathcal{E}_{A\to BW}^{\eta,\bar{n}_\textrm{B}}$. Bob receives the composite state $\hat{\rho}_{RB^n}$, where the reference system $R$ is held by Alice. He uses $\mathbf{x}$ to subselect the state $\hat{\rho}_{RB^{w(\mathbf{x})}} = \tr_{B^{\mathbf{x}^c}}(\hat{\rho}_{RB^n})$ of the non-covert QECC's output systems.

Denote by $\hat{\rho}_{RB_i}$ the state of Bob's $i^\text{th}$ subsystem in $\hat{\rho}_{RB^{w(\mathbf{x})}}$. It can be derived like $\hat{\rho}_W^{(\psi)}$ in Appendix~\ref{app:chi-square}, with careful accounting for the reference system. The state of each subsystem $B$ is demodulated by projecting it onto the qubit basis via application of $\hat{\Pi}_{B_i} = \ketbra{0}{0}_{B_i}+\ketbra{1}{1}_{B_i}$. Projection is a probabilistic process. As Alice first applies a random Pauli before transmitting, the failure probability on average is $\pfail=\frac{1+(1-\eta)\nb(3+2\nb-2\eta(\nb+1/2))}{(1+(1-\eta)\nb)^3}$.  
When Bob fails to project, he replaces his state with the maximally mixed state $\hat{\pi}_{\check{B}_i}=\frac{\hat{I}}{2}$. This yields $(1-\pfail)\hat{\Pi}_{B_i}\hat{\rho}_{RB_i}\hat{\Pi}_{B_i}^\dagger + \pfail \hat{\pi}_R\otimes\hat{\pi}_{\check{B}_i} = (1-\pfail)\bobiqubit+\pfail\hat{\pi}_R\otimes\hat{\pi}_{\check{B}_i}$, a state in the qubit basis, where $\bobiqubit$ is the projected state in the $i^\text{th}$ system. This equates to a depolarizing channel acting on $\bobiqubit$, where $\check{B}$ denotes the system after projection. Bob then uses pre-shared $\mathbf{y}$ to apply the appropriate sequence of Pauli gates, completing the Pauli twirling process. Pauli twirling by Alice and Bob on their qubit states guarantees a Pauli noise channel $\vec{q}_{\text{tw}}=(q_I,q_X,q_Y,q_Z)$ with channel parameters defined in \eqref{eq:qtw-I}-\eqref{eq:qtw-Z} and derived in Appendix~\ref{appendix:pauli-twirling}.

Recall that a Pauli channel, $\mathcal{P}_{A\to B}^{\vec{p}}$, maps input state $\hat{\rho}$ as follows:
$\mathcal{P}_{A\to B}^{\vec{p}}(\hat{\rho})=p_I \hat{I} \hat{\rho} \hat{I} + p_X \hat{X} \hat{\rho} \hat{X} + p_Y \hat{Y} \hat{\rho} \hat{Y} + p_Z \hat{Z} \hat{\rho} \hat{Z}$ for $\vec{p}=(p_I,p_X,p_Y,p_Z)$. 
A depolarizing channel with depolarizing parameter $\lambda$ is a Pauli channel given by $\mathcal{P}^{\vec{p}_\text{dep}(\lambda)}_{A\to B}(\hat{\rho}_A)$ with $\vec{p}_\text{dep}(\lambda) = \left(1-\frac{3\lambda}{4},\frac{\lambda}{4},\frac{\lambda}{4},\frac{\lambda}{4}\right)$. 
Hence, the combination of the depolarizing channel from projection to the qubit basis and Pauli channel generated by twirling yields a combined Pauli channel, $\mathcal{P}^{\vec{q}_\text{tw}}\left( \mathcal{P}^{\vec{p}_\text{dep}(\pfail)}\left(\bobiqubit\right) \right) = \mathcal{P}^{\vec{p}}_{\check{A}\to\check{B}}\left(\hat{\rho}_{R\check{A}_i}\right)$, where $\vec{p}=(p_I, p_X, p_Y, p_Z)$, $p_I = (1-\frac{3}{4}\pfail)q_I$, $p_i =(1-\frac{3}{4}\pfail)q_i+\frac{1}{4}\pfail$ for $i=X,Y,Z$, and we drop the system labels on the left hand side as they map the qubit system $\check{B}_i$ to itself. 
For every $w(\mathbf{x})\in w(\mathcal{A})$, the hashing bound guarantees the existence of a QECC with an achievable rate of $R_{\rm sr}=1-H(\vec{p})$ \cite[Sec.~24.6.3]{wilde16quantumit2ed}.

\noindent{\bf Covertness analysis:}
Willie knows the construction procedure, the channel parameters, $q$, the covert QECC, and the transmission time.
The sequence of random Pauli operations Alice applies to her encoded state is unknown to Willie. Therefore, from Willie's perspective, Alice's average non-innocent input state given $\mathbf{x}$, and tracing out the $R$ systems is $\hat{\rho}_{A^{w(\mathbf{x})}} =\mathcal{P}^{\vec{p}_\text{dep}(1)}\left(\ketbra{c({\psi})}{c({\psi)}}_{A^{w(\mathbf{x})}}\right) = \hat{\pi}^{\otimes w(\mathbf{x})},$
where $\mathcal{P}^{\vec{p}_\text{dep}(1)}(\cdot)$ is the completely depolarizing channel with $\lambda=1$ over $w(\mathbf{x})$ non-innocent input states.
Thus, Willie observes $\hat{\rho}_{W^n} \equiv \sum_{\mathbf{x}\in\mathcal{A}} p_{\mathbf{X}}\left(\mathbf{x}\right)\bigotimes_{i=1}^{n}\hat{\rho}_{W_i}^{x_i}$
where
$\hat{\rho}_{W_i}^{x_i} 
=\{\hat{\rho}_{W} \text{ if }x_i=1; \hat{\rho}_{W}^{(0)} \text{ if } x_i=0\}$, with $\hat{\rho}_{W}=\mathcal{E}^{\eta,\nb}_{A\to W}\left(\hat{\pi}\right)$ the non-innocent output state.

We upper-bound the QRE between Willie's output $\hat{\rho}_{W^n}$ and the innocent state $\hat{\rho}_{W^n}^{(0)}$ as follows:
\begin{align}
D\left(\hat{\rho}_{W^n} \middle\|\hat{\rho}_{W^n}^{(0)}\right)
&=   D\left( \left(\hat{\bar{\rho}}_{W}\right)^{\otimes n}\middle\|\hat{\rho}_{W^n}^{(0)}\right)+o(1)  \label{eq:Mehrdad-bound}\\
    &= n D\left( \hat{\bar{\rho}}_{W}\middle\|\hat{\rho}^{(0)}_W\right)+o(1)  \label{eq:QRE-additivity}\\
&\le q^2n D_{\chi^2}\left(\hat{\rho}_{W}\middle\|\hat{\rho}^{(0)}_W\right), \label{eq:chi2-bound}
\end{align}
where $\hat{\bar{\rho}}_{W}=(1-q)\hat{\rho}^{(0)}_{W} + q \hat{\rho}_{W}$, \eqref{eq:Mehrdad-bound} is from the adaptation in Appendix~\ref{app:sparse-signaling-covertness} of the covertness analysis for sparse signaling approach from 
\cite{tahmasbi21signalingcovert},
\eqref{eq:QRE-additivity} is from the additivity of the QRE over product states \cite[Ex.~11.8.7]{wilde16quantumit2ed}, and
\eqref{eq:chi2-bound} is by \cite[Lemma 1]{bullockCovertCommunicationClassicalQuantum2023} and \cite[Eq.~(9)]{ruskai1990convexity} for large enough $n$. Thus, the right-hand side of the covertness requirement 
is bounded by \eqref{eq:chi2-bound}, and Alice maintains covertness by choosing $q\leq \sqrt{2}c_\text{cov}\sqrt{\frac{\delta_{\rm C}}{n}}$, where $c_\text{cov}$ is in \eqref{eq:ccov} and $D_{\chi^2}\left(\hat{\rho}_{W}\middle\|\hat{\rho}^{(0)}_W\right)$ is derived in Appendix~\ref{app:chi-square}. 
\end{IEEEproof}
\begin{remark}The proof of this lemma has been generalized in  \cite{anderson2025covert-isit} to arbitrary channels under appropriate assumptions on the channel and Willie’s output state. These assumptions apply to the bosonic channel and are used directly in Lemma \ref{lemma:sr-no-locc} proof.
\end{remark}

\begin{lemma}\label{lemma:dr-no-locc}
Using a dual-rail photonic qubit, for $n$ large enough and arbitrary $\vartheta>0$, at least $(1-\vartheta)\sqrt{n}\frac{c_{\mathrm{cov}}}{\sqrt{2}}R_{\rm dr}\sqrt{\delta_{\rm C}}$ ebits can be transmitted reliably and covertly over $n$ bosonic channel uses, where $c_\mathrm{cov}$ is in \eqref{eq:ccov}, and $\delta_{\rm C}$ is the covertness constraint. $R_{\rm dr}\geq\left[1-H(\vec{p})\right]^+$ is the constant achievable rate of reliable ebit transmission per round, where $\vec{p}=\left[1-\frac{3p}{4},\frac{p}{4},\frac{p}{4},\frac{p}{4}\right]$, and $p=1-\frac{\eta}{(1+(1-\eta)\nb)^4}$.
\end{lemma}
\begin{IEEEproof}
    The construction in the dual-rail setting is the same as in Fig.~\ref{fig:construction-sr-no-locc}, except that Alice's QECC uses dual-rail encoding, and  $|00\rangle\langle00|^{\otimes\frac{n-w(\mathbf{x})}{2}}$ is used for sparse coding. The proof follows directly from the proof of Lemma~\ref{lemma:sr-no-locc}, where the steps in Appendix~\ref{app:chi-square} are performed on the arbitrary dual-rail input state (see \cite[Lem.~3]{anderson2024covert-qce} for the corresponding bound on $\chi^2$-divergence). The additional $\frac{1}{\sqrt{2}}$ term arises because the analysis is performed on a single photon occupying two modes, and, hence, two channel uses. Calculating the Pauli twirling parameters is not required as successful projection onto the dual-rail basis yields a depolarized state.
\end{IEEEproof}

\begin{remark} Pauli twirling is often used to approximate noise in non-Pauli channels. Indeed, we use this to transform an arbitrary quantum channel into a Pauli channel for reliability analysis. However, here, Pauli twirling also benefits the covertness analysis, ensuring that, on average, the non-innocent input state appears maximally mixed to Willie. 
\end{remark}

\section{Performance Analysis}\label{sec:performance}
\begin{figure*}[htb]
\centering
\subfloat[$\eta = 0.95$]{%
    \includegraphics[width=0.31\textwidth]{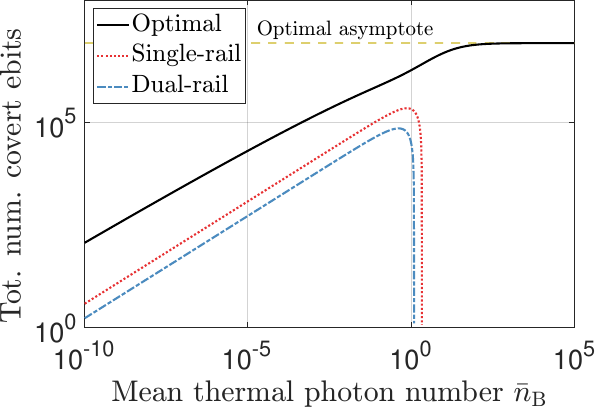}
    \label{fig:fixed-eta-rate-1}
}
\hfill
\subfloat[$\eta = 0.8$]{%
    \includegraphics[width=0.31\textwidth]{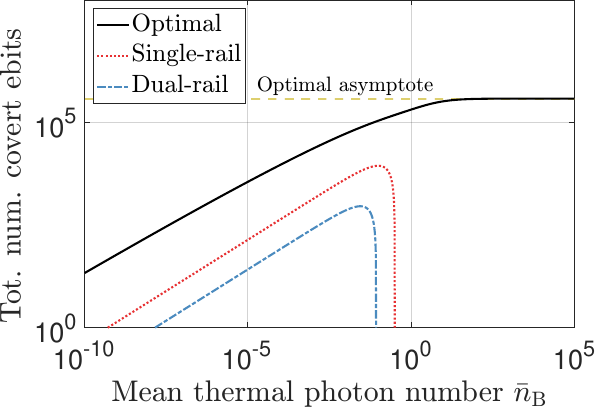}
    \label{fig:fixed-eta-rate-2}
}
\hfill
\subfloat[$\eta = 0.65$]{%
    \includegraphics[width=0.31\textwidth]{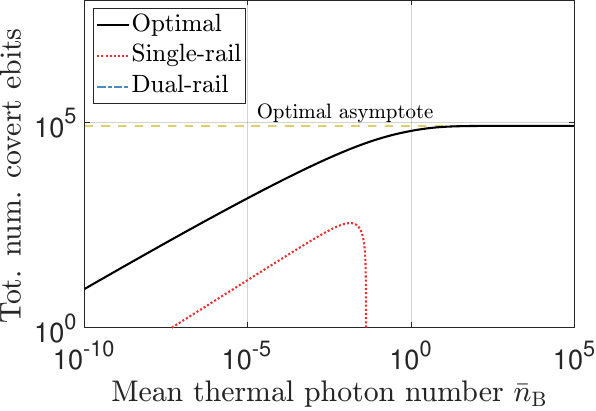}
    \label{fig:fixed-eta-rate-3}
}
\caption{Total number of covert ebits as a function of $\nb$ for a) $\eta = 0.95$, b) $\eta = 0.8$, and c) $\eta = 0.65$. 
In each subfigure, $n = 10^8$ and $\delta = 0.05$. 
The black line is the optimal rate from Theorem~\ref{thm:covert-capacity}, while the dotted red and dash-dotted blue lines correspond to the single- and dual-rail rates. 
The yellow dashed line shows the asymptotic convergence of the optimal rate as $\nb \to \infty$.}
\label{fig:fixed-eta-rates}
\end{figure*}

In Fig.~\ref{fig:fixed-eta-rates}, we plot the bounds for the total number of reliably-generated covert ebits vs.~mean thermal photon number $\nb$ for $\delta = 0.05$, $n=60\times 10^9$, and various values of transmittance $\eta$.
The black line is the optimal rate from Theorem~\ref{thm:covert-capacity}, while the dotted red and dash-dotted blue lines are the single- and dual-rail strategies. The yellow dashed line represents the limit of the optimal rate as $\nb\to \infty$. The single- and dual-rail strategies perform poorly relative to the optimal scheme, with the gap widening as $\eta$ decreases. Additionally, Fig.~\ref{fig:fixed-eta-rates} shows that there is a transmittance-dependent maximum thermal photon number beyond which neither single- nor dual-rail schemes work, motivating further exploration of practical schemes to close the gap.

\begin{figure*}[htb]
\centering
\subfloat[$\nb = 10^{-6}$]{%
    \includegraphics[width=0.31\textwidth]{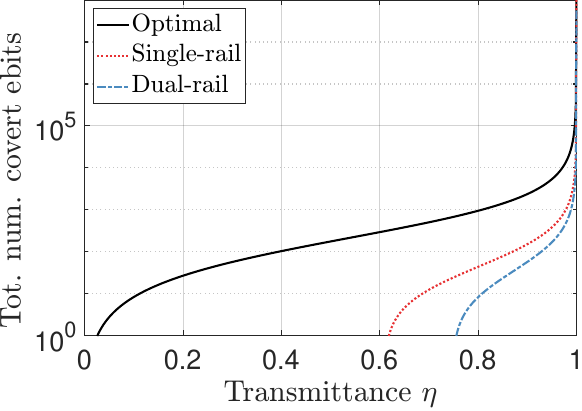}
    \label{fig:fixed-nb-rate-1}
}
\hfill
\subfloat[$\nb = 10^{-3}$]{%
    \includegraphics[width=0.31\textwidth]{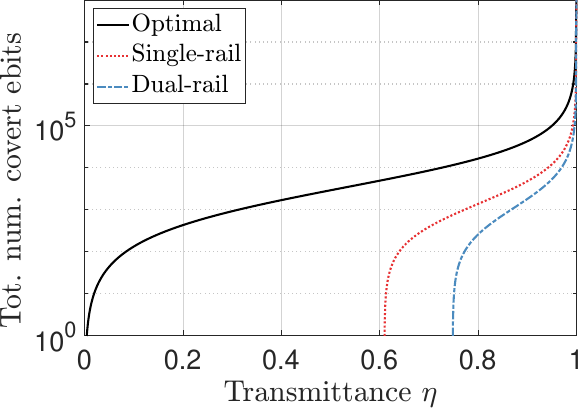}
    \label{fig:fixed-nb-rate-2}
}
\hfill
\subfloat[$\nb = 10^{-1}$]{%
    \includegraphics[width=0.31\textwidth]{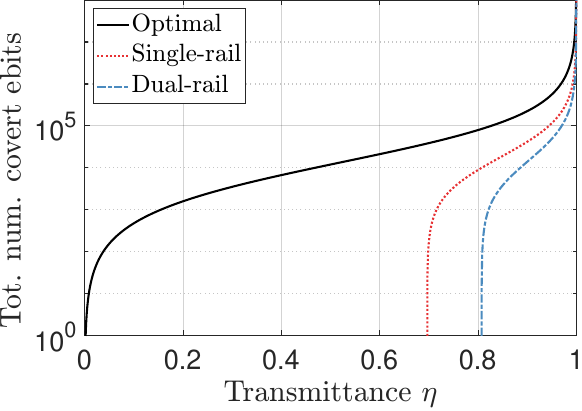}
    \label{fig:fixed-nb-rate-3}
}
\caption{Total number of covert ebits as a function of $\eta$ for a) $\nb=10^{-6}$, b) $\nb=10^{-3}$, and c) $\nb=10^{-1}$, where $n=60\times10^9$ and $\delta=0.05$ for each. 
The black line is the optimal rate from Theorem~\ref{thm:covert-capacity}, while the dotted red and dash-dotted blue lines correspond to the single- and dual-rail rates.}
\label{fig:fixed-nb-rates}
\end{figure*}

In Fig.~\ref{fig:fixed-nb-rates}, we plot the bounds for the total number of reliably generated covert ebits versus transmittance $\eta$ for $\delta = 0.05$, $n=60\times 10^9$, and various values of mean thermal photon number $\nb$. 
The black line is the optimal rate from Theorem~\ref{thm:covert-capacity}, while the dotted red and dash-dotted blue lines are the single- and dual-rail strategies. For low transmittance, the single- and dual-rail rates fail to generate any entangled pairs, whereas the optimal rate remains non-zero for all values of $\eta$. Furthermore, even for high transmittance, a near-order-of-magnitude gap persists; this gap remains even in the limit of $\eta\to1$.

\begin{figure}[htb]
\centering
\includegraphics[width=\columnwidth]{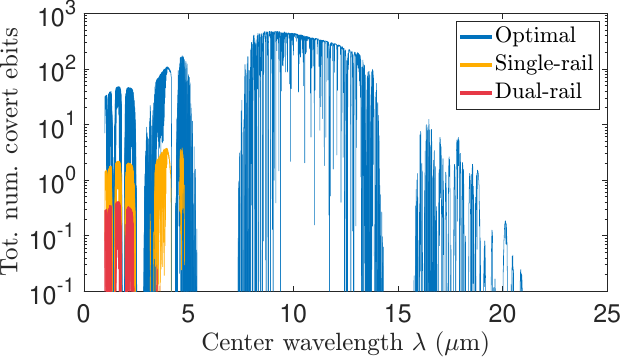}
\caption{The total number of covert ebits that can be reliably transmitted with a bandwidth of $W=10$ GHz and time $T=60$ seconds vs.~transmission center wavelength over a free-space optical link between Alice and Bob described in Section~\ref{sec:performance}. The blue line is the optimal bound from Theorem~\ref{thm:covert-capacity}. The yellow and red lines represent the single- and dual-rail strategies.}
\label{fig:modtran}
\end{figure}

In Fig.~\ref{fig:modtran} we plot the bounds for the total number of reliably generated covert ebits (entangled qubit pairs) for a free-space optical (FSO) link between Alice and Bob, who have an $L=1$ km line-of-sight channel, and aperture radii $r=10$ cm.  We set the signaling bandwidth $W=1$ GHz. Thus, over the transmission duration of $T$ seconds, we have $n=TW=10^{9}T$ optical modes. We set $T=60$s and $\delta =0.05$.  As in \cite{bash15covertbosoniccomm}, we employ a detailed MODTRAN `Mid-Latitude Summer (MLS)' atmospheric model \cite{berk06MODTRAN} for an FSO channel at elevation $10$ m above the ground level with visibility $23$ km in clear weather. We obtain the mean background thermal noise photon number from the total radiance at $60^\degree$ solar elevation.
Unlike \cite{bash15covertbosoniccomm}, we only consider the fundamental transverse electromagnetic (TEM00) spatial mode. Higher-order spatial modes do not substantially improve covert classical communication \cite[Fig.~5]{bash15covertbosoniccomm}, and, as loss increases with higher spatial-mode order, we anticipate the performance gains to be even more modest for covert entanglement generation. Per the MODTRAN model, the transmittance $\eta$ and thermal-noise photon number $\nb$ depend on the center wavelength. The blue line plots the covert capacity given in Theorem~\ref{thm:covert-capacity} while the red and orange lines correspond to the single- and dual-rail rates. We see only a small range, corresponding to high transmittance, where the single- and dual-rail strategies generate a nonzero number of ebits.

Fig.~\ref{fig:fixed-eta-rates}, Fig.~\ref{fig:fixed-nb-rates}, and Fig.~\ref{fig:modtran} show a substantial gap between the optimal rate and the rates achievable with practical strategies.  There are many potential reasons for this. First, although the channel is infinite-dimensional, we do not utilize the additional degrees of freedom. It is possible that qudit or bosonic codes could increase the achievable rates. Indeed, Gottesman-Kitaev-Preskill (GKP) qubits \cite{gottesmanEncodingQubitOscillator2001} have been shown to reach a constant factor gap to the capacity of the pure-loss channel, and perform well in the lossy thermal-noise setting \cite{noh2018capacitybounds}. However, their high photon number per mode may challenge covertness: note that the optimal scheme in Section \ref{subsec:cov-eg} spreads a small amount of power across many modes. Additionally, our strategies rely on Pauli twirling, which, by the quantum data processing inequality \cite[Ch.~10.7.2]{wilde16quantumit2ed}, can only decrease the mutual information between Alice and Bob. Lastly, to ensure that the hashing bound can be used, when projection fails for a given state, we replace it with a maximally mixed state, discarding information about where the error occurred. This can instead be treated as an erasure, aiding in decoding.

\section{Conclusion} \label{sec:conclusion} 
We explore covert entanglement generation over the bosonic channel. We obtain a single-letter expression for the optimal covert entanglement generation capacity $L_{\rm EG}$ and find that it is the same as the classical covert capacity over the bosonic channel derived in \cite{gagatsos20codingcovcomm}.  
Further, we investigate covert entanglement generation using photonic qubits and find substantial room for improvement.
This and much more is left for future investigation. 

While both our achievable covert secrecy and entanglement generation schemes have key requirements that scale as $\mathcal{O}\left(\sqrt{n}\right)$ bits, we find that the latter has a larger scaling factor when Bob has a better channel than Willie. Intuitively, this is caused by our need to fully decouple Bob's decoded states from both Willie and the environment when generating entanglement. However, it is unclear whether the extra key is necessary. To resolve this, we need a converse result for the key requirement. Adapting the approach in \cite{kimelfeld25ceg} for finite-dimensional channels to the bosonic channel is a promising path towards this. Further, our proposed practical schemes using single- and dual-rail qubits require $\mathcal{O}(\sqrt{n}\log n)$ bits due to sparse coding. We will seek to improve this in the future. 

Our single- and dual-rail results apply in the limit of asymptotically large random stabilizer codes. This requires Alice and Bob to store large numbers of physical qubits in quantum memories and employ large Clifford circuits. While we anticipate success with high-rate QECCs, e.g., quantum low-density parity-check (LDPC) codes~\cite{breuckmann2021quantum}, covert entanglement generation also needs to be investigated under practical constraints on quantum memory and circuit sizes.

\section*{Acknowledgement}
The authors are grateful to Mehrdad Tahmasbi for providing the details on the sparse coding analysis in \cite{tahmasbi21signalingcovert}, which formed the basis for Appendix \ref{app:sparse-signaling-covertness}. The authors also benefited from discussions with Matthieu R.~Bloch, Christos N.~Gagatsos, Brian J.~Smith, Ryan Camacho, Narayanan Rengaswamy, Kenneth Goodenough, and Saikat Guha.

\appendices

\section{Reliable, Covert, and Secret Random Coding Scheme}\label{ap:cov-sec-bosonic}
\begin{IEEEproof}[Proof: Lemma \ref{lem:secrecyach}]
Alice and Bob employ the construction in the proof of Theorem \ref{thm:covert-capacity-secrecy}. 
Thus, a combination of Lemma \ref{lem:PBCSQ} in Appendix \ref{ap:oneshotproof}, \cite[Cor. 11]{wilde19secondorderqkd} and  \cite[Eq. (12)]{gagatsos20codingcovcomm} gives us $
        E_\mathcal{C}\left[\bar{P}_{\rm e}(\mathcal{C})\right]\leq\epsilon$ and 
        $E_{\mathcal{C}}\left[\frac{1}{2}\left\|\hat{\bar{\rho}}^m_{W^n}-\hat{\breve{\rho}}^{\otimes n}_{W}  \right\|\right]\leq\kappa-\gamma_2, \text{ for every $m$,}$ 
    if
    \begin{align}
        \log |\mathcal{M}| &\geq n D(\hat{\rho}_{X B}\|\hat{\rho}_X\otimes\hat{\rho}_B)\nonumber \\ &\phantom{\leq} + \sqrt{nV(\hat{\rho}_{X B}\|\hat{\rho}_X\otimes\hat{\rho}_B)}\Phi^{-1}(\epsilon-C^{(1)}_n) \\
        \log |\mathcal{K}| &\leq n D(\hat{\rho}_{X W}\|\hat{\rho}_X\otimes\hat{\rho}_{W}) \nonumber \\ &\phantom{\leq}- \sqrt{nV(\hat{\rho}_{X W}\|\hat{\rho}_X\otimes\hat{\rho}_{W})}\Phi^{-1}\left(\frac{\kappa^2}{10}- C^{(2)}_n\right)
    \end{align}
    where $C^{(1)}_n = \gamma_1 + \frac{C_{\rm{BE}}T(\hat{\rho}_{X B}\|\hat{\rho}_X\otimes\hat{\rho}_B)}{\sqrt{nV(\hat{\rho}_{X B}\|\hat{\rho}_X\otimes\hat{\rho}_B)}}$ and $C^{(2)}_n = \gamma_2 + \gamma_3+ \frac{C_{\rm{BE}}T(\hat{\rho}_{X W E}\|\hat{\rho}_X\otimes\hat{\rho}_{WE})}{\sqrt{nV(\hat{\rho}_{X W E}\|\hat{\rho}_X\otimes\hat{\rho}_{WE})}}$ and $C_{\rm{BE}}$ is the Berry-Esseen constant. Now, setting $\epsilon=e^{-\varsigma_n^{(1)\sqrt{n}}}$ for $\varsigma_n^{(1)}\in \omega\left(\frac{1}{\sqrt{n}}\right)\cap o\left(1\right)$, $\kappa=e^{-\varsigma_n^{(3)}\sqrt{n}}$ for $\varsigma_n^{(3)}\in o(1)\cap \omega\left(\frac{1}{\sqrt{n}}\right)$, and $\gamma_1 \in o \left(e^{-\varsigma_n^{(1)}\sqrt{n}}\right)$ and $\gamma_2 = \gamma_3 \in o\left(e^{-\varsigma_n^{(3)}\sqrt{n}}\right)$, we have by \cite[Lem.~1, Lem.~2, Thm.~1]{gagatsos20codingcovcomm}, \cite[Eq.~(9)]{wang22isitcoverttd}, $\varsigma_n\in o(1)$, and for $n$ large enough 
    \begin{align}
            \log |\mathcal{M}| &\geq(1-\varsigma_n) n D(\hat{\rho}_{X B}\|\hat{\rho}_X\otimes\hat{\rho}_B)\\
            &= (1-\varsigma_n) n \eta\log\left(1+\frac{1}{(1-\eta)\bar{n}_b}\right)\ns\\
    \log |\mathcal{K}|&\leq (1+\varsigma_n)n D(\hat{\rho}_{X W }\|\hat{\rho}_X\otimes\hat{\rho}_{W})\\
		      &\leq (1+\varsigma_n)n\left(1-\eta\right)\log\left(1+\frac{1}{\eta \bar{n}_b}\right)\ns\label{eq:keybound}
    \end{align}
    while $         E_\mathcal{C}\left[\bar{P}_{\rm e}(\mathcal{C})\right] \leq e^{-\varsigma_n^{(1)}\sqrt{n}}$ and 
    \begin{align}
         \max_m E_\mathcal{C} \left[\frac{1}{2}\left\|\hat{\bar{\rho}}^m_{W^n}-\hat{\breve\rho}_{W}^{\otimes n}\right\|\right]&\leq e^{-\varsigma_n^{(3)}\sqrt{n}}\label{eq:semanticsecpf}.
    \end{align}
    Now, note that 
\begin{align}
    \nonumber E_\mathcal{C}&\left[\left|D(\hat{\bar\rho}_{W^n}\|\hat{\breve\rho}_W^{\otimes n})-D(\hat{\breve\rho}_W^{\otimes n}\|\hat{\rho}_{W^n}^{(0)})\right|\right]\\
    &= E_\mathcal{C}\left[\left|\left(S\left(\hat{\breve{\rho}}_{W^n}^{\otimes n}\right)-S\left(\hat{\bar\rho}_{W^n}\right)\right)\right.\right.\nonumber\\&\phantom{=E_C[|]}\left.\left.+\tr\left(\left(\hat{\breve{\rho}}_{W^n}^{\otimes n}-\hat{\bar\rho}_{W^n}\right)\log\left(\hat{\rho}_{W^n}^{(0)}
\right)\right)\right|\right].\label{eq:covbnd1}
\end{align}

Let $E = \max\left(n\ns,n\bar{n}_b\right) \in \mathcal{O}(n)$. 
By \cite[Lem.~15]{winter16cont},
\begin{align}
        \left|S\left(\hat{\breve \rho}_{W}^{\otimes n}\right)-S\left(\hat{\bar\rho}_{W^n}\right)\right|&\leq T_W n S_{\rm max}\left(\frac{2E}{n T_W}\right) + h\left( \frac{T_W}{2}\right) \nonumber\\ 
        &\in \mathcal{O}\left(\left(T_Wn\right)^2\log\left(1/T_W\right)\right)\label{eq:Scontinuity}
\end{align}
where, $T_W\equiv\left\|\hat{\bar \rho}_{W^n}-\hat{\breve{\rho}}^{\otimes n}_W \right\|$, $S_{\rm max}(E)< \infty$ is the maximum entropy under energy constraint $E<\infty$ and \eqref{eq:Scontinuity} follows from \cite[Rem.~13]{winter16cont}, \cite[Prop.~1(ii)]{Shirokov06entropychar}. Thus, \eqref{eq:semanticsecpf} implies
$    E_{\mathcal{C}}\left[\left|S\left(\hat{\breve \rho}_{W}^{\otimes n}\right)-S\left(\hat{\bar\rho}_{W^n}\right)\right|\right] \leq e^{-\gamma_n^{(2)}\sqrt{n}}
$
for some $\gamma_n^{(2)}\in o(1)\cap\omega\left(\frac{1}{\sqrt{n}}\right)$.
Finally, we examine the second term in the expectation on the right-hand side of \eqref{eq:covbnd1}:

\begin{align}
\tr\left(\hat{\Delta}_W\log\left(\hat{\rho}_{W^n}^{(0)}
\right)\right)&=\tr\left(\hat{\Delta}_W\sum_{i=1}^n\log\left(\frac{e^{-\beta(\eta \bar{n}_b) \hat{N}_{W_i}}}{Z(\beta(\eta \bar{n}_b))}\right)\right)\nonumber\\
&=\tr\left(\hat{\Delta}_W\right) n\log\left(\frac{1}{Z(\beta(\eta \bar{n}_b))}\right) \nonumber \\ &\phantom{=}-\beta(\eta \bar{n}_b) \tr\left[\sum_{i=1}^n\hat{\Delta}_W\hat{N}_{W_i}\right]\nonumber\\
&=0. \label{eq:samepower}
\end{align}
where $\hat{\Delta}_W\equiv\hat{\breve \rho}_{W}^{\otimes n}-\hat{\bar\rho}_{W^n}$ and \eqref{eq:samepower} follows from $\hat{\bar \rho}_{W^n}$ being a QPSK-modulated coherent state codeword with mean photon number per mode equal to that of $\hat{\breve \rho}_{W}^{\otimes n}$.  
\end{IEEEproof}

\section{Derandomization of Covert Secret Coding Scheme}\label{ap:lemderandom}
\begin{IEEEproof}[Proof: Lemma \ref{lem:expurge}]
    We follow the standard derandomization procedure. Define the events
    $A_{1} \triangleq \left\{\bar{P}_{\rm e}\leq e^{-\zeta_n^{(1)}\sqrt{n}}\right\}$,
    $A_{2} \triangleq \left\{\left|D(\hat{\bar\rho}_{W^n}\|\hat{\breve\rho}_W^{\otimes n}) - D(\hat{\breve\rho}_W^{\otimes n}\|\hat{\rho}_{W^n}^{(0)})\right| \leq e^{-\zeta_n^{(2)}\sqrt{n}}\right\}$, and 
    $A_{3} \triangleq \left\{\frac{1}{|\mathcal{M}|} \sum_{m} \frac{1}{2}\left\|\hat{\bar{\rho}}^m_{W^n} - \hat{\breve\rho}_{W}^{\otimes n}\right\|_1 \leq e^{-\zeta_n^{(3)}\sqrt{n}}\right\}$.
Now, Lemma \ref{lem:secrecyach}, the union bound, and Markov's inequality imply
\begin{align}
    \Pr\left(\bigcap_{i=1}^3 A_i\right) &\geq 1 - \sum_{i=1}^3 Pr(A_i)
    \geq 1 - \sum_{i=1}^3 e^{-\left(\varsigma^{(i)}_n - \zeta_n^{(i)}\right)\sqrt{n}}  \label{eq:markovs}
\end{align}
Thus, for some $\epsilon\in(0,1)$, setting $\zeta^{(k)}_n=(1-\epsilon)\varsigma^{(k)}_n$ for $k=\{1,2,3\}$ yields the right hand side of \eqref{eq:markovs} converging to unity in the limit of $n\to \infty$. Therefore, we conclude that there exists at least one coding scheme for $n$ sufficiently large that satisfies 
\begin{align}
       \bar{P}_{\rm e}(\mathcal{C}) &\leq e^{-\zeta_n^{(1)}\sqrt{n}}\label{eq:detrel} \\
       \left|D(\hat{\bar\rho}_{W^n}\|\hat{\breve\rho}_W^{\otimes n})-D(\hat{\breve\rho}_W^{\otimes n}\|\hat{\rho}_{W^n}^{(0)})\right| \label{eq:detcov}&\leq e^{-\zeta_n^{(2)}\sqrt{n}}\\
    \frac{1}{|\mathcal{M}|}\sum_{m}\left\|\hat{\bar{\rho}}^m_{W^n}-\hat{\breve\rho}_{W}^{\otimes n}\right\|_1&\leq e^{-\zeta_n^{(3)}\sqrt{n}}\label{eq:detsec}
\end{align}
Lemma \ref{lem:expurge} follows by applying \eqref{eq:detrel}-\eqref{eq:detsec} to the proof of \cite[Lem.~10]{kimelfeld25ceg}.
\end{IEEEproof}

\section{Reliable and Covert Entanglement Generation}\label{ap:entanglementgeneration}
\begin{IEEEproof}[Proof: Lemma \ref{lem:entanglementgeneration}]
We employ the strategy described in the proof of \cite[Thm.~2]{kimelfeld25ceg}. The complete analysis of the strategy for the bosonic channel is presented below for completeness. \\
\noindent{\bf State Approximation:} Consider the global state in \eqref{eq:globaldecoded}. We now state a lemma that allows us to approximate our global output state following \cite[Thm.~2]{kimelfeld25ceg}:
\begin{lemma}\label{lem:adaptlem14}\cite[Lem.~14]{kimelfeld25ceg}
    For $\ket{\tau}_{RMB^nW^nE^n\check{M}}$ defined in \eqref{eq:globaldecoded}, there exist phases $\{g(m\oplus k_2,k_1)\}$ such that $\ket{\tau}_{RMB^nW^nE^n\check{M}}$ can be approximated by 
    \begin{align}
            &\ket{\chi}_{RMB^nW^nE^n\check{M}\check{K}_1\hat{K_2}}\nonumber\\
    &=\frac{1}{\sqrt{|\mathcal{M}||\mathcal{K}_1||\mathcal{K}_2|}}\sum_{m,k_1,k_2}\ket{m}_M\otimes \ket{m}_R\nonumber\\&\phantom{==}\otimes e^{ig(m\oplus k_2,k_1)}\ket{x^n(m\oplus k_2,k_1)}_{B^nW^nE^n}
    \otimes\ket{m}_{\hat M}\nonumber \\
    &\phantom{==} \otimes \ket{k_1}_{\check{K}_1}\otimes \ket{k_2}_{\check{K}_2}\label{eq:approxdecoded}
    \end{align}
    where $\ket{x^n(m,k)}_{B^nW^nE^n}$ $ = V^{\eta,\bar{n}_b\phantom{}^\otimes n}_{A\to BWE}(\ket{x_{\rm coh}^n(m,k)}_{A^n})$
    and while $\|\ketbra{\tau}{\tau}_{RMB^nW^nE^n\check{M}}-\ketbra{\chi}{\chi}_{RMB^nW^nE^n\check{M}}\|_1\leq 2\sqrt{2}e^{-\varsigma_n^{(1)}\sqrt{n}}$.
\end{lemma}
The proof of \cite[Lem.~14]{kimelfeld25ceg} found in \cite[Appendix A]{kimelfeld25ceg} must be modified since our codewords are coherent states which are not orthogonal. Orthogonality is used in an intermediate step in the proof of \cite[Lem.~A.0.3]{wilde16quantumit2ed}. However, this proof, in fact, holds for well-behaved non-orthogonal states, as shown for completeness in Appendix \ref{ap:parsevals}.
\\
\noindent{\bf Decoupling:}
Note that we can re-express \eqref{eq:approxdecoded} by altering the system ordering as 
\begin{align}
    &\ket{\chi}_{RM\check{M}B^nW^nE^n\check{K}_1\hat{K_2}} \nonumber \\
    &=\frac{1}{\sqrt{|\mathcal{M}|}}\sum_m \ket{m}_R \otimes \ket{m}_M \otimes\ket{m}_{\check{M}}\otimes\ket{\chi^m}_{B^nW^nE^n\check{K}_1\hat{K_2}},
\end{align}
where 
\begin{align}
    &\ket{\chi^m}_{B^nW^nE^n\check{K}_1\hat{K_2}} \nonumber \\
    &= \frac{1}{\sqrt{|\mathcal{K}_1||\mathcal{K}_2|}}\sum_{k_1,k_2} e^{i g(m\oplus k_2,k_1)}\ket{x^n(m\oplus k_2, k_1)}_{B^nW^nE^n}\nonumber\\ &\phantom{=------=}\otimes\ket{k_1}_{\check{K}_1}\otimes\ket{k_2}_{\check{K}_2}
\end{align}
Bob applies the isometry $\Gamma^m_{\check{K}_2\to\check{K}_2}$ that maps $\ket{k_2}\to\ket{m\oplus k_2}$:
\begin{align}
    &\left(\hat{I}_{B^nW^nE^n\check{K}_1}\otimes \Gamma^m_{\check{K}_2\to\check{K}_2}\right) \ket{\chi^m}_{B^nW^nE^n\check{K}_1\hat{K_2}}\nonumber \\ 
    &= \frac{1}{\sqrt{|\mathcal{K}_1||\mathcal{K}_2|}}\sum_{k_1,k_2} e^{i g(m\oplus k_2,k_1)}\ket{x^n(m\oplus k_2, k_1)}_{B^nW^nE^n}\nonumber\\ &\phantom{=------=}\otimes\ket{k_1}_{\check{K}_1}\otimes\ket{m\oplus k_2}_{\check{K}_2}\\
    &= \frac{1}{\sqrt{|\mathcal{K}_1||\mathcal{K}_2|}}\sum_{k_1,k_2} e^{i g(k_2,k_1)}\ket{x^n( k_2, k_1)}_{B^nW^nE^n}\nonumber\\ &\phantom{=------=}\otimes\ket{k_1}_{\check{K}_1}\otimes\ket{k_2}_{\check{K}_2},\\
    &=\ket{\chi}_{B^nW^nE^n\check{K}_1\check{K_2}}
\end{align}
which no longer depends on $m$. Thus,  
\begin{align}
&\left(\hat{I}\otimes\Gamma^m_{\check{K}_2\to\check{K}_2}\right)\ket{\chi}_{RM\check{M}B^nW^nE^n\check{K}_1\hat{K_2}} \nonumber \\
&\phantom{=========}=\ket{\rm GHZ} \otimes \ket{\chi}_{B^nW^nE^n\check{K}_1\check{K_2}},
\end{align}
and Bob can simply trace out the $B^n\check{K}_1\check{K}_2$ systems.

\noindent{\bf Obtaining Bipartite Entanglement Without Assistance:} As in \cite{kimelfeld25ceg}, we convert the GHZ state into bipartite entanglement using a classical link and then show that the classical link is not required. We follow the same procedure as \cite[Secs.~VI.B.4, VI.B.5]{kimelfeld25ceg} with appropriate substitution of reliability bounds. \\
\noindent{\bf Covertness:}
Covertness follows from \eqref{eq:detmodcovexp} in  Lemma \ref{lem:detmodsec}.
\end{IEEEproof}

\section{One-shot Secrecy Coding Scheme}\label{ap:oneshotproof}
\begin{lemma}\label{lem:PBCSQ} \cite[Lem.~V.1]{wang22isitcoverttd}
    Consider bipartite classical-quantum state $\hat{\rho}_{XA}$ and a channel $\mathcal{N}_{XA\to XBW}$. Then, for constants $\epsilon \in (0,1)$, $\kappa\in(0,\delta/2)$, $\gamma_1\in\left(0, \frac{\epsilon^2}{10}\right)$, $\gamma_2\in\left(0, \frac{\kappa}{2}\right)$, and $\gamma_3\in\left(0, \frac{\kappa}{2}-\gamma_2\right)$, there exists a coding scheme such that 
    \begin{align}
        \log |\mathcal{M}| &\geq D_{H}^{\epsilon^2/10-\gamma_2}(\hat{\rho}_{XB}\|\hat{\rho}_X\otimes \hat{\rho}_B)-\log\frac{4\epsilon^2}{10 \gamma_1^2}, \\
        \log |\mathcal{K}|&\leq D_\mathrm{max}^{\kappa/2-\gamma_2-\gamma_3}(\hat{\rho}_{XW}\|\hat{\rho}_X\otimes \hat{\rho}_W)+2\log\frac{2\sqrt{2}}{\gamma_2\gamma_3} \label{eq:keybound1}
    \end{align}
    that satisfies $
        E_{\mathcal{C},K}\left[\bar{P}^K_{\rm e}(\mathcal{C)}\right]\leq \frac{\epsilon^2}{10}$ and $
        \max_m E_{\mathcal{C}}\left[\frac{1}{2}\left\|\hat{\bar{\rho}}^m_W-\hat{\breve{\rho}}_W  \right\|_1\right] \leq\kappa-\gamma_2$,
    where $D_{H}^{\epsilon}(\cdot\|\cdot)$ is the hypothesis testing relative entropy \cite[Defn.~1]{salek20oneshotcapacity}, $D_\mathrm{max}^{\epsilon}(\cdot\|\cdot)$ is the smooth-max relative entropy \cite[Defn.~4]{datta09minmaxrelativeentropies}, $\hat{\bar{\rho}}^m_W$ is Willie's average output state under codebook $\mathcal{C}$ conditioned on the message $m$, and $\hat{\rho}_{XW}=\tr_{B}\left[\mathcal{N}_{XA\to XBW}(\hat{\rho}_{A})\right]$, $\hat{\rho}_{XB}=\tr_{W}\left[\mathcal{N}_{XA\to XBW}(\hat{\rho}_{A})\right]$,   $\hat{\breve{\rho}}_W = \hat{\rho}_W=\tr_{XB}\left[\mathcal{N}_{XA\to XBW}(\hat{\rho}_{A})\right]$.
    \end{lemma}
\begin{IEEEproof}
The proof follows from construction and analysis in \cite[App.~B]{wang22isittowardtdjournal}. Here, we have increased the key length in \eqref{eq:keybound1} compared to \cite[Eq.~(6)]{wang22isitcoverttd}, allowing for the strengthened secrecy bound compared to that of \cite[Lemma V.1]{wang22isitcoverttd}. 
\end{IEEEproof}

\section{Non-orthogonal Coherent POVM}\label{ap:parsevals}
Here, we modify the proof of \cite[Lemma A.0.3]{wilde16quantumit2ed} to show that it holds for non-orthogonal collections of states $\{\ket{\zeta_i}\}$ and $\{\ket{\chi_i}\}$. We use notation therein. It suffices to show that the first equality in \cite[Eq. (A.19)]{wilde16quantumit2ed} holds for these: $\frac{1}{N}\sum_s\braket{\hat{\chi}_s|\hat{\zeta}_s} = \frac{1}{N^2}\sum_{s,k,l}e^{\frac{j2\pi(l-k) s}{N}}\braket{{\chi}_k|{\zeta}_l} 
    =\frac{1}{N}\sum_i\braket{{\chi}_i|{\zeta}_i}$.

\section{Covertness Analysis for Sparse Signaling}
\label{app:sparse-signaling-covertness}
Here we adapt the approach from 
\cite{tahmasbi21signalingcovert} 
to show validity of \eqref{eq:Mehrdad-bound}. By the definition of QRE, we have:
\begin{align}
    \nonumber &\left|D\left(\hat{\rho}_{W^n} \middle\|\left(\hat{\rho}^{(0)}_W\right)^{\otimes n}\right)-D\left( \hat{\bar{\rho}}_{W}^{\otimes n}\middle\|\left(\hat{\rho}^{(0)}_W\right)^{\otimes n}\right)\right|\\
    &= \left|\left(S\left(\hat{\bar{\rho}}_{W}^{\otimes n}\right)-S\left(\hat{\rho}_{W^n} \right)\right)\right. \notag
    \\&\phantom{=|}\left.+\tr\left(\left(\hat{\bar{\rho}}_{W}^{\otimes n}-\hat{\rho}_{W^n} \right)\log\left(\hat{\rho}_{W^n}^{(0)}\right)
\right)\right|.\label{eq:qre-equality-decomposition}
\end{align}
where $S(\hat{\rho})=-\tr\left(\hat{\rho}\log\hat{\rho}\right)$ is the von Neumann entropy of quantum state $\hat{\rho}$.
We upper-bound the last two terms of \eqref{eq:qre-equality-decomposition}.
First, let $\epsilon\equiv \frac{1}{2}\left\|\left(\hat{\bar{\rho}}_{W}\right)^{\otimes n}-\hat{\rho}_{W^n} \right\|_1,$ and note that
\begin{align}
    \epsilon
    \label{eq:data_processing}
    &\leq\frac{1}{2}\sum_{\mathbf{x}}\left|\prod_{i=1}^{n}p_X(x_i)-p_{\mathbf{X}}(\mathbf{x})\right| 
    \\
    &=p\left(\bar{\mathcal{A}}\right) 
     \leq 2 e^{-\frac{1}{3} qn\vartheta^2},\label{eq:chernoff}
\end{align}
where \eqref{eq:data_processing} is by the data processing inequality with classical statistical (total variation) distance, $p\left(\bar{\mathcal{A}}\right)=1-p\left(\mathcal{A}\right)$,
and the inequality in \eqref{eq:chernoff} is the Chernoff bound.
For our setting of $q\propto 1/\sqrt{n}$, $\epsilon$ decays to zero exponentially in $\sqrt{n}$.

Let $\hat{H}_{W_i}=\hbar\omega(\hat{n}_{W_i}+\frac{1}{2})$ be the Hamiltonian for the $i^{\text{th}}$ system at Willie where $\hat{n}_{W_i}$ is the number operator, and let $
E\equiv \max\left(\sum_{i=1}^n\tr\left(\hat{\bar{\rho}}_{W}\hat{H}^{W_i}\right),\sum_{i=1}^n\tr\left(
\hat{\rho}_{W^n} \hat{H}^{W_i}\right)\right)=O(n)$,
where $f(n)=O(g(n))$ means that $f(n)$ grows no faster asymptotically than $g(n)$: $\limsup_{n\to\infty} \left|\frac{f(n)}{g(n)}\right|<\infty$. By \cite[Lemma 15]{winter16cont},
\begin{align}
        &S\left(\hat{\bar{\rho}}_{W}^{\otimes n}\right)-S\left(\hat{\rho}_{W^n} \right)\leq2\epsilon n S_{\rm max}\left(\frac{E}{\epsilon n}\right) + h(\epsilon), \label{eq:infinite-dim-fannes}
\end{align}
where $S_{\rm max}(E)< \infty$ is the maximum entropy under energy constraint $E<\infty$. \cite[Rem.~13]{winter16cont}, \cite[Prop.~1(ii)]{Shirokov06entropychar} and the fact that $\epsilon n\to 0$ implies that \eqref{eq:infinite-dim-fannes} vanishes as $n\to\infty$. 

To bound the last term in \eqref{eq:qre-equality-decomposition}, we decompose $\hat{\bar{\rho}}_{W}^{\otimes n}=p\left(\mathcal{A}\right)\hat{\rho}_{W^n} +p\left(\bar{\mathcal{A}}\right)\label{eq:prob_not_in_A}\hat{\sigma}_{W^n}$,
where $\hat{\sigma}_{W^n}$ is a density operator with $E_{\hat{\sigma}}\equiv\sum_{i=1}^n\tr\left(\hat{\sigma}_{W^n}\hat{H}_{W_i}\right)=O(n)$.
Let $\hat{\Delta}_{W^n}\equiv \hat{\rho}_{W^n} -\hat{\sigma}_{W^n}$ and note that $\tr\left(\hat{\Delta}_{W^n}\right)=0$ and $E_{\hat{\Delta}}\equiv\sum_{i=1}^n\tr\left(\hat{\Delta}_{W^n}\hat{H}_{W_i}\right)=O(n)$.
Then,
\begin{IEEEeqnarray}{rCl}
\IEEEeqnarraymulticol{3}{l}
{\tr\left(\left(\hat{\bar{\rho}}_{W}^{\otimes n}-\hat{\rho}_{W^n} \right)\log\left(\hat{\rho}_{W^n}^{(0)}
\right)\right)}\IEEEnonumber\\
&=&p\left(\bar{\mathcal{A}}\right)\tr\left(\hat{\Delta}_{W^n}\log\left(\hat{\rho}_{W^n}^{(0)}
\right)\right)\\
&=&p\left(\bar{\mathcal{A}}\right)\tr\left(\hat{\Delta}_{W^n}\sum_{i=1}^n\log\left[\frac{1}{Z(\beta(E_0))}e^{-\beta(E_0) \hat{H}_{W_i}}\right]\right)\label{eq:operatorexp}\IEEEeqnarraynumspace\\
&=&p\left(\bar{\mathcal{A}}\right)n\log\left(\frac{1}{Z(\beta(E_0))}\right)\tr\left(\hat{\Delta}_{W^n}\right) -p\left(\bar{\mathcal{A}}\right)\beta(E_0)E_{\hat{\Delta}}\notag\\
&=&o(1)\label{eq:qre_equality_fourth_term},\
\end{IEEEeqnarray}
where $\eqref{eq:operatorexp}$ is due to operator exponentiation and that $\hat{\rho}_W^{(0)}=\frac{1}{Z(\beta(E_0))}e^{-\beta(E_0) \hat{H}_W}$ is a thermal state with $\hat{H}_W=\hat{n}_w$, and \eqref{eq:qre_equality_fourth_term} follows from \eqref{eq:chernoff}.

\section{Pauli Twirling Channel Parameters} \label{appendix:pauli-twirling} 
Alice implements Pauli twirling by choosing a random Pauli operator and applying it to the qubit state $\hat{\rho}_A$ before transmission. Bob applies the same operator to ``undo'' the twirling at his output, yielding: $\hat{\rho}_B = \frac{1}{4}\sum_{\hat{P}\in\mathcal{Q}}\hat{P}\mathcal{N}(\hat{P}\hat{\rho}_A\hat{P})\hat{P}$, 
where $\mathcal{Q}\equiv\{\hat{I},\hat{X},\hat{Y},\hat{Z}\}$ is the Pauli operator set.

Given a channel with Kraus operators $\hat{K}_j$ indexed by $i$, the Choi state of the channel is $\hat{\rho}^\text{Choi} = \sum_{i=0}^\infty (\hat{I}\otimes \hat{K}_i)\ket{\Phi^+}\bra{\Phi^+}(\hat{I}\otimes \hat{K}_i)$,
and the Pauli channel parameters generated by Pauli twirling in terms of the Choi state are:
\begin{align}
    p_I &= \bra{\Phi^+} \hat{\rho}^\text{Choi} \ket{\Phi^+} \label{eq:pauli-I} , 
    p_X = \bra{\Psi^+} \hat{\rho}^\text{Choi} \ket{\Psi^+} , \\
    p_Y &= \bra{\Psi^-} \hat{\rho}^\text{Choi} \ket{\Psi^-} ,
    p_Z = \bra{\Phi^-} \hat{\rho}^\text{Choi} \ket{\Phi^-} ,\label{eq:pauli-Z}
\end{align}
where $\ket{\Phi^{\pm}}\equiv\frac{1}{\sqrt{2}}(\ket{00}\pm\ket{11})$ and $\ket{\Psi^{\pm}}\equiv\frac{1}{\sqrt{2}}(\ket{01}\pm\ket{10})$ are the Bell states.

Let us now apply this to the lossy thermal-noise channel followed by projection. The lossy thermal-noise channel can be decomposed into a pure-loss channel of transmissivity $\tau=\eta/G$ followed by a quantum-limited amplifier with gain coefficient $G=1+(1-\eta)\nb$. As we are primarily concerned with the Choi state, the decomposition has the following Kraus operator representation with input state $\hat{\rho}_{RA}=\ketbra{\Phi^+}{\Phi^+}_{RA}$: $\mathcal{E}^{\eta,\nb}_{A\to B}(\hat{\rho}_{RA}) = \sum_{k,l} \hat{B}_k \hat{A}_l \hat{\rho}_{RA}\hat{A}_l^\dagger \hat{B}_k^\dagger$
where $\hat{A}_l=\hat{I}_R\otimes \sqrt{\frac{(1-\tau)^l}{l!}} \tau^{\frac{\hat{n}}{2}} \hat{a}^l$ and 
$\hat{B}_k =\hat{I}_R\otimes \sqrt{\frac{1}{k!} \frac{1}{G}\left(\frac{G-1}{G}\right)^k} \hat{a}^{\dagger k} G^{-\frac{\hat{n}}{2}}$.
 $\hat{a}, \hat{a}^\dagger$ are the annihilation and creation operators for system $A$, and $\hat{n}=\hat{a}^\dagger\hat{a}$ is the number operator.

Furthermore, Bob projects his physical channel output state onto the computational basis, where the projection operator is $\hat{\Pi}=\hat{I}_R\otimes\ketbra{0}{0}_B+\hat{I}\otimes\ketbra{1}{1}_B$. This has failure probability $\pfail$. When Bob fails to project, he replaces the state with the maximally mixed state $\hat{\pi}$.
The total action of the physical channel and projection yields:
\begin{align}
    \hat{\rho}_{RB}=(1-\pfail)\hat{\Pi}\left(\sum_{k=0}^\infty\sum_{l=0}^\infty \hat{B}_k \hat{A}_l \hat{\rho}_{RA}\hat{A}_l^\dagger \hat{B}_k^\dagger\right)\hat{\Pi}^\dagger + \pfail\hat{\pi}_{RB} \label{eq:rhob-choi}.
\end{align}
By definition, $\hat{\rho}_{RB}$ in \eqref{eq:rhob-choi} defines the Choi state for our channel with projection when $\hat{\rho}_{RA} = \hat{\rho}^{(\Phi^+)}_{RA}=\ket{\Phi^+}\bra{\Phi^+}_{RA}$. Since twirling $\hat{\pi}_{RB}$ yields the same state, we only need to determine the parameters for the state $\hat{\rho}_{RB}^{\textrm{Choi}}$ conditioned on projection success. Let us further simplify \eqref{eq:rhob-choi}. The system's initial state in the Bell basis limits $l$'s range to $\{0,1\}$, since the annihilation operator applied more than once will yield the vacuum state. Furthermore, the projection operator $\hat{\Pi}$ only selects states within the Bell basis, indicating that the creation operator from $\hat{B}_k$ can only support $k=0,1$. Hence,
\begin{align}
   \hat{\rho}_{RB}^{\textrm{Choi}} &=  \hat{\Pi}\left(\sum_{k,l\in\{0,1\}}\hat{B}_k \hat{A}_l \hat{\rho}^{(\Phi^+)}_{RA}\hat{A}_l^\dagger \hat{B}_k^\dagger\right)\hat{\Pi}^\dagger \\
    &=\frac{1}{2 N(G,\tau)}\left[ \frac{1}{G}\ketbra{00}{00}+\frac{G-1}{G^2}\ketbra{01}{01}\notag \right . \\
    &\phantom{=}+\frac{1-\tau}{G}\ketbra{10}{10}+\frac{G+\tau-1}{G^2}\ketbra{11}{11} \notag \\
    &\phantom{=}+\left. \frac{\sqrt{\tau}}{G^{\frac{3}{2}}}\left(\ketbra{00}{11}+\ketbra{11}{00} \right) \right], \label{eq:rhob-choi2}
\end{align}
where $ N(G, \tau) = \frac{G(2-\tau)+\tau-1}{2G^2}$ is a normalization term.
Note that terms containing bi-photon states $\ket{12}$ are removed by the projection. Evaluating \eqref{eq:pauli-I} and \eqref{eq:pauli-Z} for \eqref{eq:rhob-choi2} yields the twirling parameters in \eqref{eq:qtw-I}-\eqref{eq:qtw-Z}.

\section{\texorpdfstring{$\chi^2$}{chi-squared}-divergence Derivation for Single-rail Qubits}\label{app:chi-square}
We first derive the output state at Willie to calculate the $\chi^2$-divergence between Willie's observed innocent state and when Alice transmits. In the single-rail case, Alice inputs an arbitrary qubit into the channel of the form $\hat{\rho}_A = \begin{pmatrix} |\alpha|^2 & \gamma \\ \gamma^*& |\beta|^2 \end{pmatrix}$.

Utilizing anti-normally ordered characteristic functions,
\begin{align}
    \chi_A^{\hat{\rho}_W}(\zeta) &= \chi_A^{\hat{\rho}_{A}}\left(\sqrt{1-\eta}\zeta\right) \chi_A^{\hat{\rho}_E}\left(-\sqrt{\eta}\zeta\right),
\end{align}
where $\hat{\rho}_W^{(\psi)}$ is the state at Willie, $\hat{\rho}_E$ is the environment's thermal state input. For a thermal state with mean thermal photon number $\nb$, the characteristic function is known  \cite[Sec.~7.4.3.2]{orszag16quantumotpics}:
\begin{align}
\chi_A^{\hat{\rho}_E}(-\sqrt{\eta}\zeta) = e^{-(1+\nb)\eta|\zeta|^2},
\end{align}

The next step is finding the characteristic function for $\hat{\rho}_A$.
\begin{align}
    &\chi_A^{\hat{\rho}_A}\left(\sqrt{1-\eta}\zeta\right)  \notag \\
    &= e^{-(1-\eta)|\zeta|^2} \tr\left(\hat{\rho}_A e^{\zeta \sqrt{1-\eta} \hat{a}^\dagger} e^{-\zeta^* \sqrt{1-\eta} \hat{a}}  \right) \label{eq:BCHtheorem}\\
    &=e^{-(1-\eta)|\zeta|^2}\sum_{n=0}^\infty \bra{n} \left( |\alpha|^2 \ketbra{0}{0}+\gamma \ketbra{0}{1}
+\gamma^*\ketbra{1}{0} \notag \right. \\
    &\phantom{e^{-(1-\eta)|\zeta|^2}\sum_{n=0}^\infty}\left .+\,|\beta|^2\ketbra{1}{1} \right) e^{\zeta \sqrt{1-\eta} \hat{a}^\dagger} e^{-\zeta^* \sqrt{1-\eta} \hat{a}}  |n\rangle. \label{eq:char-a}
\end{align}
We  evaluate \eqref{eq:char-a} term by term, noting that the infinite sums drop due to the orthogonality of photon number states. Let us first evaluate the $|\alpha|^2$-term:
\begin{align}
    &|\alpha|^2\bra{0} e^{\zeta \sqrt{1-\eta} \hat{a}^\dagger} e^{-\zeta^* \sqrt{1-\eta} \hat{a}} \ket{0} \\
    &=|\alpha|^2\bra{0}\left(\sum_{k^{\prime}=0}^n(\hat{a}^\dagger)^{k^\prime}\frac{(\sqrt{1-\eta} \zeta)^{k^{\prime}}}{k^\prime!} \right) \notag \\
    &\phantom{=}\quad\times\left(\sum_{k=0}^n\frac{\left(-\sqrt{1-\eta} \zeta^*\right)^k}{k!} (\hat{a})^k\right) \ket{0} \\
    &= |\alpha|^2, \label{eq:alpha-eval}
\end{align}
where $\eqref{eq:alpha-eval}$ due to the fact that the summands are not zero only when the ladder operators are not applied, i.e., $k^\prime = k = 0$.

Evaluation of the $|\beta|^2$ term follows that of the $|\alpha|^2$ term:
\begin{align}
    &|\beta|^2\bra{1} e^{\zeta \sqrt{1-\eta} \hat{a}^\dagger} e^{-\zeta^* \sqrt{1-\eta} \hat{a}} \ket{1} \\
    &= |\beta|^2 \left(1 - |\zeta|^2(1-\eta)\right). \label{eq:beta-eval}
\end{align}
Furthermore, 
\begin{align}
    \gamma\bra{1} e^{\zeta \sqrt{1-\eta} \hat{a}^\dagger} e^{-\zeta^* \sqrt{1-\eta} \hat{a}} \ket{0} &=  \gamma \zeta\sqrt{1-\eta} \label{eq:gamma-eval} \\
    -\gamma^*\bra{0} e^{\zeta \sqrt{1-\eta} \hat{a}^\dagger} e^{-\zeta^* \sqrt{1-\eta} \hat{a}} \ket{1} &= -\gamma^* \zeta^*\sqrt{1-\eta} \label{eq:gamma-star-eval}
\end{align}
Combining these terms,  \eqref{eq:char-a} evaluates to:
\begin{align}
 &\chi_A^{\hat{\rho}_A}\left(\sqrt{1-\eta}\zeta\right)  \notag \\
    &=
e^{-(1-\eta)|\zeta|^2}\left(|\alpha|^2+ |\beta|^2 \left(1 - |\zeta|^2(1-\eta)\right)\vphantom{\gamma \zeta\sqrt{1-\eta} -\gamma^* \zeta^*\sqrt{1-\eta}} \right. \notag \\
&\phantom{=}+ \left.\gamma \zeta\sqrt{1-\eta} -\gamma^* \zeta^*\sqrt{1-\eta}  \right).
\end{align}

A quantum state $\hat{\rho}_W$ describes Willie's observation. It is related to its characteristic function $\chi_A^{\hat{\rho}_W}(\cdot)$ via the operator Fourier transform \cite{weedbrook12gaussianQIrmp}:
\begin{align}
    \hat{\rho}_W^{(\psi)} &= \int \frac{d^2\zeta}{\pi}\chi_A^{\hat{\rho}_W}(\zeta) e^{\zeta\hat{w}^\dagger} e^{-\zeta^*\hat{w}}, \label{eq:fourier-transform}
\end{align}
where the integral is over the complex plane for $\zeta$. Converting to polar coordinates using $|\zeta|^2 = r^2$, $\zeta = r(\cos(\theta)+\im \sin(\theta))$, $\zeta^* = r(\cos(\theta)-\im \sin(\theta))$, and complex unit $\im\equiv\sqrt{-1}$ yields:
\begin{align}
    \hat{\rho}_W^{(\psi)} &= \int_{r=0}^\infty r dr\int_{\theta=0}^{2\pi} \frac{d\theta}{\pi}\chi_A^{\hat{\rho}_W} (r,\theta)e^{r\left(\cos(\theta)+\im \sin(\theta)\right)\hat{w}^\dagger} \notag \\
    &\phantom{=}\qquad\qquad\qquad\times e^{-r(\cos(\theta)-\im \sin(\theta))\hat{w}}. 
\end{align}
Likewise, $\chi_A^{\hat{\rho}_{W^{(\psi)}}}$ in polar coordinates is:
\begin{align}
    \chi_A^{\hat{\rho}_{W^{(\psi)}}}(r, \theta) &= e^{-(1+\eta\nb)\eta r^2} \notag\\ 
    &\phantom{=}\times\left( 1 - |\beta|^2r^2(1-\eta) \right. \notag \\
    &\phantom{=\times}+ r\sqrt{1-\eta}(\gamma (\cos(\theta)+\im \sin(\theta))) \notag \\
    &\phantom{=\times}\left.-r\sqrt{1-\eta}(\gamma^* (\cos(\theta)-\im \sin(\theta))) \right),
\end{align}
where $|\alpha|^2+|\beta|^2=1$ removes $\alpha$.
Furthermore, a generalized single-mode state in the photon number basis is:
\begin{align}
    \hat{\rho}_W^{(\psi)} = \sum_{m=0}^\infty\sum_{n=0}^\infty \bra{m}\hat{\rho}_W^{(\psi)} \ket{n}\ket{m}\bra{n}, \label{eq:gen-state}
\end{align}
where we replace the right-hand-side $\hat{\rho}_W^{(\psi)}$ with that of \eqref{eq:fourier-transform} and evaluate the integrals to get the following state for $\hat{\rho}_W^{(\psi)}$:
\begin{align}
    \hat{\rho}_W^{(\psi)} &= \sum_{m=0}^\infty \left[\frac{(\eta\nb)^m}{(1+\eta\nb)^{m+1}} \notag \right.\\
    &\phantom{=}+\left. |\beta|^2(1-\eta)\frac{(\eta\nb)^{m-1}(m-\eta\nb)}{(1+\eta\nb)^{m+2}}\right ] \ketbra{m}{m} \notag \\
    &\phantom{=}-\gamma  \sqrt{(1-\eta)(m+1)}\frac{(\eta\nb)^m}{(1+\eta\nb)^{m+2}} \ketbra{m}{m+1} \notag \\
    &\phantom{=}-\gamma^*  \sqrt{(1-\eta)m}\frac{(\eta\nb)^{m-1}}{(1+\eta\nb)^{m+1}}\ketbra{m}{m-1}. \label{eq:rhow}
\end{align}

As Alice applies Pauli twirling, Willie observes, on average, a maximally mixed state with $\alpha = \beta = 1/\sqrt{2}$ and $\gamma=\gamma^*=0$. Then $\eqref{eq:rhow}$ yields $\hat{\rho}_W$ used in the proof of Lemma~\ref{lemma:sr-no-locc}:
\begin{align}
\hat{\rho}_W &= \sum_{m=0}^\infty \left[\frac{(\eta\nb)^m}{(1+\eta\nb)^{m+1}} \right. \notag \\
&\phantom{=}\left. + \frac{(1-\eta)}{2}\frac{(\eta\nb)^{m-1}(m-\eta\nb)}{(1+\eta\nb)^{m+2}}\right ] \ketbra{m}{m}. \label{eq:willies-output-state}
\end{align}
Willie's output state when Alice is quiet is an attenuated thermal state $\hat{\rho}_W^{(0)} = \hat{\rho}_{\rm th}(\eta\nb)$ defined in \ref{sec:systemmodel}.

Since both $\hat{\rho}_W$ and $\hat{\rho}_W^{(0)}$ are diagonal in the photon number basis, calculation of $D_{\chi^2}\left(\hat{\rho}_W\middle\|\hat{\rho}_W^{(0)}\right) = \operatorname{tr}[(\hat{\rho}_W)^2(\hat{\rho}_W^{(0)})^{-1}] - 1$ is straightforward. We square the coefficients of $\hat{\rho}_W$ and take the reciprocal of $\hat{\rho}_W^{(0)}$'s coefficients, and multiply them term by term. Taking the infinite sum of the elements, and using known identities from \cite{gr07tables}, yields $D_{\chi^2}\left(\hat{\rho}_W\middle\|\hat{\rho}_W^{(0)}\right) = \frac{(1-\eta)^2}{4\eta\nb(1+\eta\nb)}$.

\bibliographystyle{IEEEtran}
\bibliography{papers}

\end{document}